\documentclass[12pt]{article}
\usepackage[utf8]{inputenc}
\usepackage[sectionbib]{natbib}

\title{Tensor-Train AR}

\usepackage[margin=0.85in]{geometry}
\usepackage{graphicx}
\usepackage{mathrsfs}
\usepackage{multirow}
\usepackage{multicol}
\usepackage{amsfonts}
\usepackage{amsmath}
\usepackage{comment}
\usepackage{subfig}
\usepackage{amsthm}
\usepackage{mathtools}
\usepackage{algorithm}
\usepackage{algpseudocode}
\usepackage{sectsty}
\usepackage{makecell}
\usepackage{amsmath}
\usepackage{amssymb}
\usepackage[toc,page]{appendix}
\usepackage[mathscr]{euscript}
\usepackage[toc]{appendix}

\let\counterwithin\relax
\usepackage{chngcntr}
\usepackage{apptools}
\usepackage{booktabs}
\usepackage{caption}
\usepackage{lscape}
\usepackage{tabularx} %
\usepackage[utf8]{inputenc} %
\usepackage{array} %
\usepackage{arydshln}
\usepackage{soul}
\usepackage{tikz}
\usetikzlibrary{shapes.misc}
\usetikzlibrary{shapes}
\usepackage{underscore} 
\AtAppendix{\counterwithin{lemma}{section}}
\usepackage{indentfirst}
\usepackage{epstopdf}
\usepackage{diagbox}

\newtheorem{assumption}{Assumption}

\newtheorem{remark}{Remark}
\newtheorem{proposition}{Proposition}
\newtheorem{theorem}{Theorem}

\newcolumntype{P}[1]{>{\centering\arraybackslash}p{#1}} %
\algnewcommand{\Initialize}[1]{%
	\State \textbf{Initialize:}
	\Statex \hspace*{\algorithmicindent}\parbox[t]{.8\linewidth}{\raggedright #1}
}
\newcommand{\norm}[1]{\left\lVert#1\right\rVert}

\usepackage[mathscr]{euscript}
\usepackage{mathtools}
\usepackage{amsthm}
\usepackage{amsfonts}
\usepackage{amssymb}
\usepackage{subfig}
\usepackage{color}
\usepackage{multirow}
\usepackage{makecell}
\usepackage{booktabs}
\usepackage{tikz}
\newcommand{\bm}{\boldsymbol}
\newcommand{\cm}[1]{\mbox{\boldmath$\mathscr{#1}$}}
\newcommand{\Lrateten}{\raisebox{0pt}{\tikz{\draw[blue,solid,line width = 1.0pt,fill=blue](2.75mm,0.75mm) circle (0.7mm);\draw[-,blue,solid,line width = 1.5pt](0.,0.8mm) -- (5.5mm,0.8mm)}}}
\newcommand{\Lratetwenty}{\raisebox{0pt}{\tikz{\node[draw=green,scale=0.1,minimum size=5mm,rotate=0,line width=1.mm] at (2.75mm,0.65mm) {};\draw[-,green,dashed,line width = 1.5pt](0.,0.65mm) -- (5.5mm,0.65mm)}}}

\newcommand{\Lratefive}{\raisebox{-2pt}{\tikz{\node[draw=red,scale=0.35,cross out,minimum size=5mm,rotate=0,line width=1.mm] at (2.75mm,0.65mm) {};\draw[-,red,densely dotted,line width = 1.5pt](0.,0.65mm) -- (5.5mm,0.65mm)}}}

\DeclareMathOperator*{\vectorize}{vec}
\DeclareMathOperator*{\rank}{rank}

\DeclareMathOperator*{\argmin}{arg\,min}

\definecolor{blue}{rgb}{0.0,0.0,1.0}
\definecolor{green}{rgb}{0.0,0.5,0.0}
\definecolor{red}{rgb}{1.0,0.0,0.0}
\definecolor{cyan}{rgb}{0.0,1.0,1.0}

\usepackage{lipsum}
\usepackage[colorlinks,linkcolor=blue,citecolor=blue]{hyperref}

\title{An efficient tensor regression for high-dimensional data}
\author{Yuefeng Si, Yingying Zhang, Yuxi Cai, Chunling Liu and Guodong Li
	\\ \textit{University of Hong Kong, East China Normal University} \\ \textit{and Hong Kong Polytechnic University} }

\begin{document}
	\maketitle

\begin{abstract}
Many existing tensor regression models for high-dimensional data are based on Tucker decomposition, which has good properties but loses its efficiency in compressing tensors very quickly as the order of tensors increases, say greater than four or five.
However, for the simplest tensor autoregression in handling time series data, its coefficient tensor already has the order of six.
This paper revises a newly proposed tensor train (TT) decomposition and then applies it to tensor regression such that a nice statistical interpretation can be obtained.
The new tensor regression can well match the data with hierarchical structures, and it even can lead to a better interpretation for the data with factorial structures, which are supposed to be better fitted by models with Tucker decomposition.
More importantly, the new tensor regression can be easily applied to the case with higher order tensors since TT decomposition can compress the coefficient tensors much more efficiently.
The methodology is also extended to tensor autoregression for time series data, and nonasymptotic properties are derived for the ordinary least squares estimations of both tensor regression and autoregression.
A new algorithm is introduced to search for estimators, and its theoretical justification is also discussed. 
Theoretical and computational properties of the proposed methodology are verified by simulation studies, and the advantages over existing methods are illustrated by a real example.

\end{abstract}
	

\newpage
\section{Introduction}
Modern society has witnessed an enormous progress in collecting all kinds of data with complex structures, and datasets in the form of tensors have been increasingly encountered from many fields, such as signal processing \citep{zhao2012higher,Shimoda2012}, medical imaging analysis \citep{Zhou2013,li2018tucker}, economics and finance \citep{Chen2021,Wang2021}, digital marketing \citep{hao2021sparse,bi2018multilayer} and many others. 
Many unsupervised learning methods have been considered for these datasets, and they include the principal component analysis \citep{Zhang2022}, clustering \citep{Sun2019,Luo2022}, and factor modeling \citep{bi2018multilayer,Chen2021}.
On the other hand, there is a bigger literature for analyzing tensor-valued observations with supervised learning methods, and most of them come from the area of machine learning by using neural networks \citep{Novikov2015, Kossaifi2020}. 
As the most important supervised learning method, regression plays an important role in modeling statistical association between responses and predictors and then making forecasts for responses, and its application on tensor-valued observations, called { regression, has recently attracted more and more attention in the literature \citep{Raskutti2019,Han2020,Wang2021a}.

Tensor regression usually involves a huge number of parameters, which actually grows exponentially with the order of tensors, and the situation is much more serious for tensor autoregression in handling time series data since both responses and predictors are high order tensors with the same sizes \citep{Wang2021a,Chen2021a}.
To address this challenge, some dimension reduction tools will be needed for the coefficients of a tensor regression, while canonical polyadic (CP) and Tucker decomposition are two most commonly used tools for compressing tensors in the literature \citep{Kolda2009}.
CP decomposition \citep{harshman1970foundations} can compress a tensor dramatically by representing them into a linear combination of basis tensors with rank one, and this makes it very useful in handling higher order tensors.
{Accordingly, many researches in statistics employ CP decomposition to restrict the parameter space of tensor regression \citep{Zhou2013, rajarshi2017bayesian, lock2018tensor} or to factorize tensor-valued data directly \citep{dunson2009nonparam, zhou2015bayesian, lock2018supervised, bi2018multilayer}.
However, it is generally NP-hard in computation to determine the CP-rank of a tensor \citep{Hillar2013}, and the corresponding approximation with a fixed rank in the Frobenius norm can be ill-posed \citep{Kolda2009}.} As a result, its calculation is well known to be computationally unstable \citep{Kolda2009, Oseledets2010,Hillar2013}. 

In the meanwhile, {Tucker ranks and decomposition \citep{Tucker1966} can be always and easily determined, say by high order singular value decomposition, and hence its numerical performance is the most stable among almost all tensor decomposition tools \citep{Kolda2009}.}
More importantly, tensor regression with Tucker decomposition can be well interpreted from statistical perspectives; see Section \ref{adv:TT} for details.
As a result, Tucker decomposition has dominated the applications of tensor regression \citep{zhao2012higher,li2018tucker, Raskutti2019,hao2021sparse,Han2020,gahrooei2021multiple} and autoregression \citep{Wang2021a, Wang2021, li2021multilinear}.
However, Tucker decomposition is well known not be able to effectively compress higher order tensors, and {it loses the efficiency very quickly as the order of a tensor increases while the sample size is maintained at the same level} \citep{oseledetsTR}.
This hinders its further applications on tensor regression.
As an illustration, consider a time series with observed values at each time point being a third order tensor, and an autoregression with order one is applied.
The corresponding coefficient tensor has the order of six, and the number of free parameters will be larger than $5^6=15625$ even if Tucker ranks are as small as five; see Section \ref{adv:TT} for more details.
In fact, these successful empirical applications of tensor regression with Tucker decomposition are all limited to the data with very low order tensors.

Tensor train (TT) decomposition \citep{oseledetsTR} is a newly proposed dimension reduction tool for tensor compressing, and it inherits advantages from both CP and Tucker decomposition.
Specifically, TT decomposition can compress tensors as significantly as CP decomposition, while its calculation is as stable as Tucker decomposition; see \cite{oseledetsTR}. 
Due to its better performance in practice, TT decomposition has been widely used in many areas, including machine learning \citep{Novikov2015,Yang2017,Kim2019,Su2020}, physics and quantum computation \citep{orus2019tensor,bravyi2021classical}, signal processing \citep{cichocki2015tensor}, and many others.
On the other hand, it has an obvious drawback and, comparing with CP and Tucker decomposition, it is much harder to interpret the components from TT decomposition.
\cite{Zhou2020} proposed a tensor train orthogonal iteration (TTOI) algorithm to conduct low-rank TT approximation for noisy high order tensors and, as a byproduct, the relationship between the sequential matricization of a tensor and the components of its TT decomposition was also established. This opens a door to apply TT decomposition to statistical problems.

The first contribution of this paper is to revise the classical TT decomposition by introducing orthogonality to some matrices in Section \ref{sec:2-1}, and then tensor regression with the new TT decomposition can be well interpreted from statistical perspectives in Section \ref{adv:TT}.
When applying to tensor regression, both Tucker and TT decomposition first summarize responses and predictors into factors and then regress response factors on predictor ones, while CP decomposition does not have such interpretation.
On the other hand, comparing with Tucker decomposition, TT decomposition has three advantages: (i.) it exactly matches the data with nested, or hierarchical, structures since its factor extraction proceeds sequentially along modes; (ii.) even for the data with factorial or nested-factorial structures, TT decomposition will lead to much less number of factors by taking into account the coherence among modes of response or predictor tensors, and hence the resulting model can be better interpreted; and (iii.) TT decomposition regresses the response factors on predictor ones in the element-wise sense, and it hence can compress the parameter space dramatically.
\cite{Liu2020} considered the classical TT decomposition to tensor regression, while there is no interpretation and theoretical justification.
This paper also has another three main contributions:
\begin{itemize}
	\item[(a)] The ordinary least squares (OLS) estimation is considered for tensor regression with the new TT decomposition in Section \ref{sec:tt}, and its nonasymptotic properties are established. As a byproduct, the covering number for tensors with TT decomposition is also derived, and it may be of independent interest.
	\item[(b)] The new methodology is applied to time series data, resulting in a new tensor autoregressive model in Section \ref{sec:TT-AR}. Since time series data have special structures and probabilistic properties, it is nontrivial to arrange TT decomposition into the model and to derive the nonasymptotic properties.
	\item[(c)] A Riemannian gradient descent algorithm is introduced to search for estimators in Section \ref{Implementation}, and its theoretical justification is discussed. Moreover, two selection methods for TT ranks are proposed with their selection consistency being justified asymptotically.
\end{itemize}

In addition, Section \ref{sim} conducts simulation experiments to evaluate the finite-sample performance
of the proposed methodology, and its usefulness is further demonstrated by a real example in Section \ref{realdata}. A short conclusion and discussion is given in Section \ref{sec:7}, and technical proofs are provided in a separate supplementary file.

\section{Tensor train decomposition and tensor regression}\label{sec:2}
\subsection{Tensor decomposition}\label{sec:2-1}	

This subsection introduces three tensor decomposition techniques: \textit{Tucker, canonical polyadic (CP) and tensor train (TT) decomposition}; see Section A.1 of the supplementary file and \cite{Kolda2009} for the introduction of tensor notations and algebra.

For a $d$-th order tensor $\cm{X}\in\mathbb{R}^{p_1\times p_2\times\cdots \times p_d}$, denote by $[\cm{X}]_{(s)}$ and $[\cm{X}]_{s}$ the mode matricization and sequential matricization at mode $s$ with $1\leq s\leq d$, respectively.
There are two commonly used dimension reduction tools for tensors in the literature: Tucker and CP decomposition.
Suppose that, for each $1\leq i\leq d$, mode-$i$ matricization of $\cm{X}$ has a low rank, i.e. $r_i=\rank([\cm{X}]_{(i)})$, and then the tensor $\cm{X}$ is said to have Tucker ranks $(r_1,r_2,\ldots,r_d)$.
As a result, there exists a Tucker decomposition,
	\begin{equation}\label{Tucker-format}
	\cm{X}=\cm{G}\times_1\bm{U}_1\times_2\bm{U}_2\times_3 \cdots\times_d\bm{U}_d,
\end{equation}
where $\cm{G}\in\mathbb{R}^{r_1\times r_2\times\cdots\times r_d}$ is the core tensor and $\bm{U}_i\in\mathbb{R}^{p_i\times r_i}$ with $1\leq i\leq d$ are factor matrices.
Accordingly, for each $1\leq i\leq d$, its mode-$i$ matricization could be expressed as
\[
[\cm{X}]_{(i)} = \bm{U}_i[\cm{G}]_{(i)}(\bm{U}_d\otimes\cdots\otimes\bm{U}_{i+1}\otimes\bm{U}_{i-1}\otimes\cdots\otimes\bm{U}_1)^\top,
\]
where $\otimes$ denotes Kronecker product.
Note that Tucker decomposition is not unique, since
$\cm{X}=\cm{G}\times_1\bm{U}_1\times_2\cdots\times_d\bm{U}_d=(\cm{G}\times_1\bm{O}_1\times_2\cdots\times_d\bm{O}_d)\times_1 (\bm{U}_1 \bm{O}_1^{-1}) \times_2 \cdots \times_d(\bm{U}_d \bm{O}_d^{-1})$
for any invertible matrices $\bm{O}_i\in\mathbb{R}^{r_i\times r_i}$ with $1\leq i\leq d$. We can consider the higher order singular value decomposition (HOSVD) of $\cm{X}$, a special Tucker decomposition defined by choosing $\bm{U}_i$ as the tall matrix consisting of the top $r_i$ left singular vectors of $[\cm{X}]_{(i)}$ and then setting $\cm{G}=\cm{X}\times_1\bm{U}_1^\top \times_2\cdots\times_d\bm{U}_d^{\top}$.
Note that $\bm{U}_i$'s are orthonormal, i.e. $\bm{U}_i^\top \bm{U}_i=\bm{I}_{r_i}$ with $1\leq i\leq d$, and their corresponding projection matrices, $\bm{U}_i \bm{U}_i^\top$'s, can be uniquely defined.

Tucker decomposition has a wonderful performance for lower order tensors, especially for third order tensors, while it loses the efficiency very quickly as $d$ increases.
This is due to the fact that Tucker decomposition fails to compress the space dramatically \citep{oseledetsTR}. In the meanwhile, {CP decomposition is the best candidate for this scenario}, and it has the form of
\begin{equation}\label{CP-format}
	\cm{X}=\sum_{j=1}^rg_{j}\bm{u}_1(j)\circ\bm{u}_2(j)\circ\cdots \circ\bm{u}_d(j),
\end{equation}
where $\circ$ denotes the outer product, $\bm{u}_i(j)\in\mathbb{R}^{p_i}$ and $\|\bm{u}_i(j)\|_2=1$ for all $1\leq i\leq d$ and $1\leq j\leq r$, and $\bm{U}_i=(\bm{u}_i(1),\ldots,\bm{u}_i(r))\in\mathbb{R}^{p_i\times r}$ with $1\leq i\leq d$ are factor matrices.
{Not like Tucker decomposition, the factor matrices in CP decomposition can be uniquely identified under some regularity conditions \citep{Kolda2009}.
However, when establishing theoretical properties for tensor regression with coefficient tensors having CP decomposition, we may need to orthogonalize $\bm{U}_i$'s in order to count the corresponding covering number, and this will result in a Tucker decomposition in general; see Remark \ref{rem1} for more details.
As a result, its derived statistical properties may not be better than those with Tucker decomposition.} Moreover, CP decomposition is well known to be unstable numerically \citep{Oseledets2010}.

The recently proposed TT decomposition \citep{oseledetsTR} is another dimension reduction tool for tensors.
This tool is stable like Tucker decomposition, while {it can compress the space almost as significantly as CP decomposition}.
Specifically, each element of tensor $\cm{X}\in\mathbb{R}^{p_1\times p_2\times\cdots \times p_d}$ can be decomposed into
\begin{equation}\label{TT-format}
	\cm{X}_{i_1,i_2,\ldots, i_d}=\bm{G}_1(i_1)\bm{G}_2(i_2)\cdots\bm{G}_{d-1}(i_{d-1})\bm{G}_d(i_d),
\end{equation}
where $\bm{G}_{k}(i_{k})$ is an $r_{k-1}\times r_{k}$ matrix for $1\leq k\leq d$, and $r_0=r_d=1$.
For $2\leq k\leq d-1$, we stack these matrices $\bm{G}_{k}(i_{k})$'s into a tensor $\cm{G}_{k}\in\mathbb{R}^{r_{k-1}\times p_{k}\times r_{k}}$ such that $[\cm{G}_{k}]_{(3)}=(\bm{G}_k^\top(1),\ldots,\bm{G}_k^\top(p_k))$. Moreover, let $\bm{G}_{1}=(\bm{G}_{1}^\top(1),\ldots,\bm{G}_1^\top({p_1}))^\top\in\mathbb{R}^{p_{1}\times r_{1}}$ and $\bm{G}_{d}=(\bm{G}_{d}(1),\ldots,\bm{G}_d({p_d}))^\top\in\mathbb{R}^{p_{d}\times r_{d-1}}$, and we call $\bm{G}_1,\cm{G}_2 \cdots\cm{G}_{d-1}, \bm{G}_d$ the TT cores.
From \eqref{TT-format}, the element of $\cm{X}$ can be obtained by multiplying the $i_{1}$-th vector of $\bm{G}_{1}$, $i_{2}$-th matrix of $\cm{G}_{2}$, $\ldots$, $i_{d-1}$-th matrix of $\cm{G}_{d-1}$, and the $i_{d}$-th vector of $\bm{G}_{d}$ sequentially. It is like a ``train", and hence the name of tensor train decomposition. 
Moreover, the $j$-th left part of $\cm{X}$ is defined as a $\prod_{i=1}^jp_i$-by-$r_j$ matrix, denoted by $\bm{G}^{\leq j}$, and its $i_1\cdots i_j$-th row for each $1\leq i_l\leq p_l$ with $1\leq l\leq j$ has the form of
$$\bm{G}^{\leq j}(i_1\cdots i_j, :) := \bm{G}_1(i_1)\bm{G}_2(i_2)\cdots\bm{G}_j(i_j).$$
Similarly, the $j$-th right part of $\cm{X}$ is a $r_{j-1}$-by-$\prod_{i=j}^dp_i$ matrix, $\bm{G}^{\geq j}$, and its $i_j\cdots i_d$-th column is $\bm{G}^{\geq j}(:,i_1\cdots i_d) := \bm{G}_j(i_j)\bm{G}_{j+1}(i_{j+1})\cdots\bm{G}_d(i_d)$, where $1\leq j\leq d$ and $\bm{G}^{\leq0}=\bm{G}^{\geq d+1}=[1]$.

The TT ranks are defined as $(r_1,\ldots,r_{d-1})$, and it can be verified that $r_i=\rank([\cm{X}]_i)$ for all $1\leq i\leq d-1$.
From Lemma 3.1 in \cite{Zhou2020}, its sequential matricization has a form of
\[
[\cm{X}]_k=(\bm{I}_{p_2\cdots p_k}\otimes\bm{G}_1)\cdots(\bm{I}_{p_k}\otimes[\cm{G}_{k-1}]_2)[\cm{G}_k]_2[\cm{G}_{k+1}]_1
	([\cm{G}_{k+2}]_1\otimes\bm{I}_{p_{k+1}})\cdots(\bm{G}_{d}^\top\otimes\bm{I}_{p_{k+1}\cdots p_{d-1}})
\]
for each $1\leq k\leq d-1$, and we denote it by $\cm{X}=[[\bm{G}_1, \cm{G}_2, \ldots, \cm{G}_{d-1},\bm{G}_d]]$ for simplicity. Note that, not like the form at \eqref{TT-format}, this form makes it possible to statistically interpret TT decomposition.
{Similarly, TT decomposition is not unique, since
$\cm{X}=[[\bm{G}_1\bm{O}_1, \cm{G}_2\times_1\bm{O}_1^{-1}\times_3\bm{O}_2^\top, \ldots,\cm{G}_{d-1}\times_1\bm{O}_{d-2}^{-1}\times_3\bm{O}_{d-1}^\top,\bm{G}_d(\bm{O}_{d-1}^{-1})^\top]]$
for any invertible matrices $\bm{O}_i\in\mathbb{R}^{r_i\times r_i}$ with $1\leq i\leq d-1$.}
To fix the problem, this paper introduces another variant, by adding the orthogonality to some matrices, such that TT decomposition can be applied to tensor regression with nice statistical interpretation in the next subsection.

\begin{proposition}\label{prop1}
	For a fixed $1\leq k\leq d-1$, the  $k$th sequential matricization of tensor $\cm{X}$ at \eqref{TT-format} can be decomposed into
	\begin{equation}\label{TTdecomposition}
		[\cm{X}]_k=(\bm{I}_{p_2\cdots p_k}\otimes\bm{G}_1)\cdots(\bm{I}_{p_k}\otimes[\cm{G}_{k-1}]_2)[\cm{G}_k]_2\bm{\Sigma}[\cm{G}_{k+1}]_1
		([\cm{G}_{k+2}]_1\otimes\bm{I}_{p_{k+1}})\cdots(\bm{G}_{d}^\top\otimes\bm{I}_{p_{k+1}\cdots p_{d-1}}),
	\end{equation}
	where $\bm{G}_1^\top\bm{G}_1=\bm{I}_{r_1}$, $[\cm{G}_i]_2^\top[\cm{G}_i]_2=\bm{I}_{r_i}$ for $2\leq i\leq k$, $[\cm{G}_i]_1[\cm{G}_i]_1^\top=\bm{I}_{r_{i-1}}$ for $k+1\leq i\leq d-1$, $\bm{G}_d^\top\bm{G}_d=\bm{I}_{r_{d-1}}$, and $\bm{\Sigma}\in\mathbb{R}^{r_k\times r_k}$ is a diagonal weight matrix.
\end{proposition}	

The above decomposition can be conducted by Algorithm S.1 in the supplementary file.
For a tensor satisfying Proposition \ref{prop1}, we denote it by $\cm{X}=[[\bm{G}_1,\cm{G}_2,\ldots,\cm{G}_{k},\bm{\Sigma},\cm{G}_{k+1},\ldots,\cm{G}_{d-1},\bm{G}_{d}]]$ for simplicity.

In addition, this paper denotes vectors by boldface small letters, e.g. $\bm{x}$; matrices by boldface capital letters, e.g. $\bm{X}$; tensors by boldface Euler capital letters, e.g. $\cm{X}$.  For a generic matrix $\bm{X}$, we denote by $\bm{X}^\top$, $\norm{\bm{X}}_{\mathrm{F}}$, $\norm{\bm{X}}_2$, and $\vectorize(\bm{X})$ its transpose, Frobenius norm, spectral norm, and a long vector obtained by stacking all its columns, respectively. If $\bm{X}$ is further a square matrix, then denote its minimum and maximum eigenvalue by $\lambda_{\min}(\bm{X})$ and $\lambda_{\max}(\bm{X})$, respectively.
For two real-valued sequences $x_k$ and $y_k$, $x_k\gtrsim y_k$ if there exists a $C>0$ such that $x_k\geq Cy_k$ for all $k$. In addition, we write $x_k\asymp y_k$ if $x_k\gtrsim y_k$ and $y_k\gtrsim x_k$.	
Moreover, the Frobenius norm of $\cm{X}$ is defined as $\norm{\cm{X}}_{\mathrm{F}}=\sqrt{\langle\cm{X},\cm{X}\rangle}$, where $\langle\cdot,\cdot\rangle$ is the inner product.
	
\subsection{Advantages of tensor train decomposition for tensor regression}\label{adv:TT}

This subsection applies tensor train (TT) decomposition at Proposition \ref{prop1} to a simple tensor regression such that its statistical interpretation can be well demonstrated.
Its advantages are also presented from both physical interpretations and efficiency over Tucker decomposition, which has been commonly used for tensor regression \citep{Raskutti2019,Wang2021}.

Consider a simple tensor regression with matrix-valued responses and predictors,
\begin{equation}\label{simple-regression}
\bm{Y}_i = \langle\cm{A},\bm{X}_i\rangle+\bm{E}_i \hspace{2mm}\text{for}\hspace{2mm} 1\leq i\leq N,
\end{equation}
where $\bm{Y}_i\in\mathbb{R}^{q_1\times q_2}$, $\bm{X}_i\in\mathbb{R}^{p_1\times p_2}$, $\cm{A}\in\mathbb{R}^{q_1\times q_{2}\times p_1\times p_2}$  is the coefficient tensor,  $\bm{E}_i\in\mathbb{R}^{q_1\times q_2}$ is the error term, and $N$ is the sample size.
We first apply Tucker decomposition to model \eqref{simple-regression} for dimension reduction.
Suppose that the coefficient tensor $\cm{A}$ has Tucker ranks $(r_1,r_2,r_3,r_4)$, and we then have a Tucker decomposition,
$\cm{A}=\cm{G}\times_{i=1}^{4}\bm{U}_{i}$, where $\bm{U}_{i}\in\mathbb{R}^{q_i\times r_i}$ and $\bm{U}_{j}\in\mathbb{R}^{p_{j-2}\times r_j}$ with $i\in\{1,2\}$ and $j\in\{3,4\}$ are orthonormal factor matrices, and $\cm{G}\in\mathbb{R}^{r_{1}\times r_{2}\times r_3\times r_4}$ is the core tensor.
As a result, model \eqref{simple-regression} can be rewritten into
\begin{equation}\label{model-tucker}
	\bm{U}_1^{\top}\bm{Y}_i\bm{U}_2=\langle\cm{G}, \bm{U}_3^{\top}\bm{X}_i\bm{U}_4\rangle+\bm{U}_1^{\top}\bm{E}_i\bm{U}_2,
\end{equation}
where dimension reduction is first conducted to the responses $\bm{Y}_i\in\mathbb{R}^{q_1\times q_2}$ and predictors $\bm{X}_i\in\mathbb{R}^{p_1\times p_2}$, and we then regress the resulting lower dimensional response factors $\bm{U}_1^{\top}\bm{Y}_i\bm{U}_2\in\mathbb{R}^{r_1\times r_2}$ on predictor factors $\bm{U}_3^{\top}\bm{X}_i\bm{U}_4\in\mathbb{R}^{r_3\times r_4}$; see Figure B.1 in the supplementary file for the illustration.


We next conduct dimension reduction to model \eqref{simple-regression} by TT decomposition.
Suppose that $\cm{A}$ has TT ranks $(r_{1},r_{2},r_3)$, and it then has a TT decomposition $\cm{A}=[[\bm{G}_1,\cm{G}_2,\bm{\Sigma},\cm{G}_{3},\bm{G}_{4}]]$, i.e.
\[
[\cm{A}]_2=(\bm{I}_{q_2}\otimes\bm{G}_1) [\cm{G}_2]_2\bm{\Sigma}[\cm{G}_{3}]_1 (\bm{G}_{4}^\top\otimes\bm{I}_{p_{1}}),
\]
where $\bm{G}_1^\top\bm{G}_1=\bm{I}_{r_1}$, $[\cm{G}_2]_2^\top[\cm{G}_2]_2=[\cm{G}_3]_1[\cm{G}_3]_1^\top=\bm{I}_{r_{2}}$, $\bm{G}_4^\top\bm{G}_4=\bm{I}_{r_3}$, and $\bm{\Sigma}\in\mathbb{R}^{r_2\times r_2}$ is a diagonal matrix.
From model \eqref{simple-regression},
$\vectorize(\bm{Y}_i) = [\cm{A}]_2\vectorize(\bm{X}_i)+\vectorize(\bm{E}_i)$,
and it then holds that
\begin{equation}\label{model-tt}
	[\cm{G}_2]_{2}^\top\vectorize(\bm{G}_1^{\top}\bm{Y}_i)= \bm{\Sigma}[\cm{G}_{3}]_1
	\vectorize(\bm{X}_i\bm{G}_4)+[\cm{G}_2]_{2}^\top\vectorize(\bm{G}_1^{\top}\bm{E}_i),
\end{equation}
where $[\cm{G}_2]_{2}^\top\vectorize(\bm{G}_1^{\top}\bm{Y}_i)= \langle \bm{G}_1^{\top}\bm{Y}_i, \cm{G}_2 \rangle$, and $[\cm{G}_{3}]_1
\vectorize(\bm{X}_i\bm{G}_4) =\langle \cm{G}_{3},\bm{X}_i\bm{G}_4\rangle$; see Figure B.1 in the supplementary file for its illustration.
It can be seen that the responses $\bm{Y}_i\in\mathbb{R}^{q_1\times q_2}$ and predictors $\bm{X}_i\in\mathbb{R}^{p_1\times p_2}$ are first summarized into response factors $[\cm{G}_2]_{2}^\top\vectorize(\bm{G}_1^{\top}\bm{Y}_i)\in\mathbb{R}^{r_2}$ and predictor factors $[\cm{G}_{3}]_1
\vectorize(\bm{X}_i\bm{G}_4)\in\mathbb{R}^{r_2}$, respectively, and we then regress the response factors on predictor ones in the element-wise sense since the coefficient matrix $\bm{\Sigma}\in\mathbb{R}^{r_2\times r_2}$ is diagonal.

\begin{figure}[ht]
	\centering
	\includegraphics[width=0.8\linewidth,height=0.22\textheight]{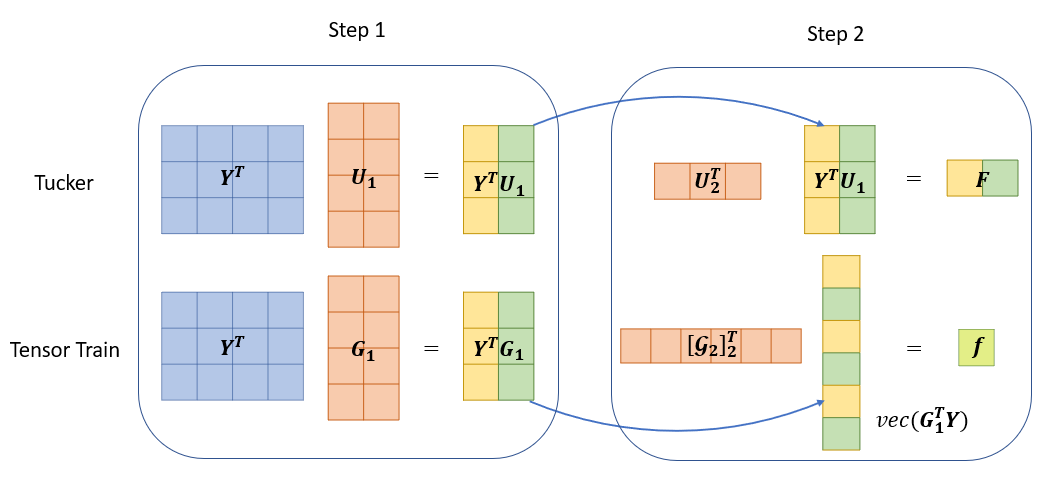}	\caption{\label{illu-factor}Illustration on how Tucker and TT decomposition extract response factors for a simple tensor regression with matrix-valued responses.}
\end{figure}

On one hand, both models via Tucker and TT decomposition share similar two-stage procedures: summarizing responses and predictors into factors, and then regressing response factors on predictor ones.
On the other hand, the two procedures are different in details. We next demonstrate how Tucker and TT decomposition conduct dimension reduction for responses $\bm{Y}$ in the first stage.
Specifically, from \eqref{model-tucker} and \eqref{model-tt}, the $q_1$ variables along the first mode of $\bm{Y}\in\mathbb{R}^{q_1\times q_2}$ are first summarized into $r_1$ factors, $\bm{U}_1^{\top}\bm{Y}_i$ and $\bm{G}_1^{\top}\bm{Y}_i$, by Tucker and TT decomposition in the same way.
When further compressing the second mode of $\bm{Y}\in\mathbb{R}^{q_1\times q_2}$, Tucker decomposition conducts dimension reduction for the $q_2$ variables within each of summarized $r_1$ factors, while TT decomposition does it for all summarized $q_2r_1$ factors; see Figure \ref{illu-factor} for the illustration.
Note that TT decomposition extracts response and predictor factors in a reverse order.
Moreover, CP decomposition does not enjoy the above factor modeling interpretation. 


TT decomposition is more suitable for the data with nested, or hierarchical, structures, and it also enjoys advantages for the data with a seemingly factorial structure by creating much less number of factors.
For example, consider the monthly international import trade data from January 2010 to December 2019, and we choose $q_1=6$ product categories and $q_2=6$ countries such that these countries have different preferred product categories for importation.
The predictors are set to $\bm{X}_i=\bm{Y}_{i-1}^\top$. TT and Tucker decomposition are applied to coefficient tensor $\cm{A}$ with ranks $(2,1,1)$ and $(2,2,1,1)$, respectively, such that they have similar estimation errors and the same number of parameters.
Note that TT decomposition leads to one response factor, while there are four factors for Tucker decomposition.
Specifically, from \eqref{model-tt}, TT decomposition first employs $\bm{G}_1\in\mathbb{R}^{6\times 2}$ to extract $r_1=2$ factors for each country from the product feature through the operation $\bm{G}_1^\top\bm{Y}_i$. From the left panel of Figure \ref{trade:tt}, the first factor represents a contrast between industrial materials and textile materials \& products, while the second one pays more attention to industrial materials.
Note that there are totally 12 factors, $\vectorize(\bm{G}_1^\top\bm{Y}_i)\in\mathbb{R}^{12}$, and they are further summarized into one response factor by using $[\cm{G}_2]_2\in\mathbb{R}^{12\times 1}$; see the middle panel of Figure \ref{trade:tt}.
The larger weight at the second column corresponds to industrial materials at USA, while those at the eleventh and twelfth columns imply the larger contribution of Brazil.
The right panel of Figure \ref{trade:tt} presents the loading matrix of the response factor, which is actually the combination of $\bm{G}_1$ and $\cm{G}_2$, and the products and countries are coherent to each other obviously.
As a result, at least four factors from Tucker decomposition are needed to obtain a similar representing power; see Figure B.3 in the supplementary file for details.

\begin{figure}[ht]
	\centering
	\includegraphics[width=\textwidth]{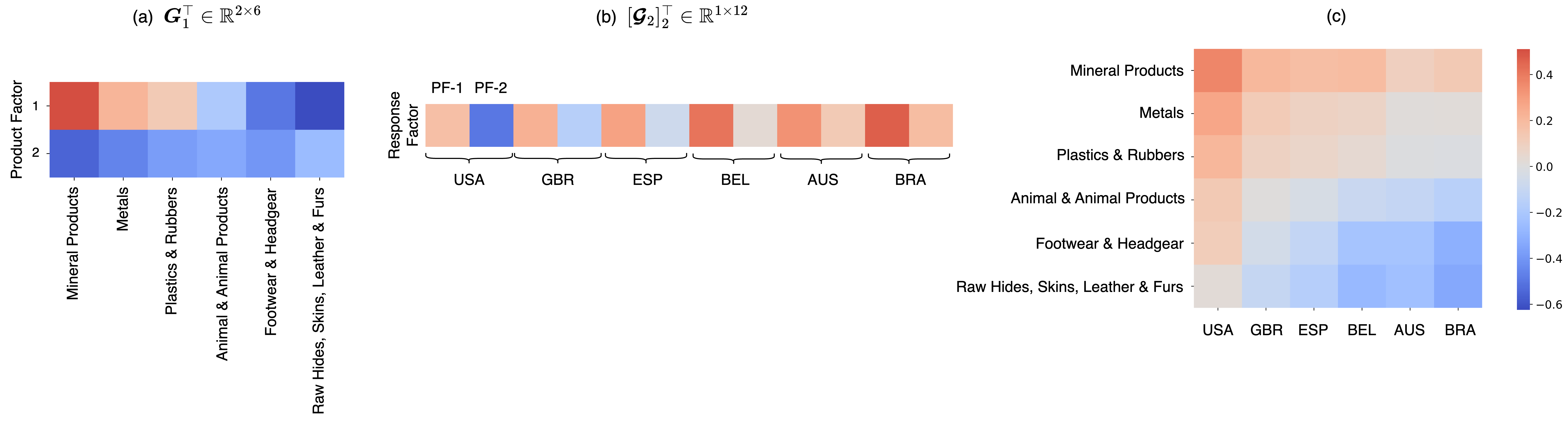}
	\caption{Heatmaps of $\bm{G}_1^\top$ (left panel), $[\cm{G}_2]_2^\top$ (middle panel) and the loading matrix (right panel) of one response factor from TT decomposition. Six countries are USA, United Kingdom (GBR), Spain (ESP), Belgium (BEL), Australia (AUS) and Brazil (BRA).}
	\label{trade:tt}
\end{figure}


Finally, the models via Tucker and TT decomposition will involve different numbers of parameters in the second stage, and they correspond to the sizes of $\cm{G}$ at \eqref{model-tucker} and $\bm{\Sigma}$ at \eqref{model-tt}, which are $\prod_{i=1}^4r_i$ and $r_2$, respectively.
Consider a model with both responses and predictors being third-order tensors, and suppose that the coefficient tensor $\cm{A}$ has Tucker ranks $(r_1,\cdots,r_6)$ with $r_i$'s being as small as five. As a result, the size of $\cm{G}$ is $5^6=15625$.
In fact, Tucker decomposition has a poor performance for a higher order tensor \citep{oseledetsTR}, and the situation is much more serious for regression problems.
In the meanwhile, we can always control the size of $\bm{\Sigma}$ within a reasonable range.
This further confirms the necessity of TT decomposition for tensor regression, especially for responses and predictors with high order tensors.

\section{Tensor train regression for high-dimensional data}\label{sec:3}
\subsection{Tensor train regression}\label{sec:tt}

This subsection considers a general regression with tensor-valued responses and predictors,
\begin{equation}\label{regmodel}
	\cm{Y}_i = \langle\cm{A},\cm{X}_i\rangle+\cm{E}_i \hspace{2mm}\text{for}\hspace{2mm} 1\leq i\leq N,
\end{equation}
where $\cm{Y}_i\in\mathbb{R}^{q_1\times q_2\times \cdots\times q_n}$, $\cm{X}_i\in\mathbb{R}^{p_1\times p_2\times\cdots\times p_m}$, $\cm{A}\in\mathbb{R}^{q_1\times\cdots\times q_n\times p_1\times\cdots\times p_m}$ is the coefficient tensor, $\cm{E}_i\in\mathbb{R}^{q_1\times q_2\times\cdots\times q_n}$ is the error term, and $N$ is the sample size.
As in Section \ref{adv:TT}, the responses $\cm{Y}_i$ are supposed to have a nested or hierarchical structure, and its $j$th mode is nested under each level of $(j+1)$th mode for all $1\leq j\leq n-1$.
In the meanwhile, it is common that some modes are arranged in a factorial layout and other modes are nested, and this corresponds to a nested-factorial design in the area of experimental designs \citep[Chapter 14]{Montgomery2009}.
In fact, all modes may even be factorial.
For these cases, the modes are first arranged according to our experiences and, for those modes with no preference, we may simply put the ones with more levels, i.e. larger values of $q_j$'s, at lower places since they will be summarized into lower dimensional factors earlier.
Finally, since the predictor factors are extracted in a reverse order, we assume that the predictors $\cm{X}_i$ also have a nested structure, while its $j$th mode is nested under each level of $(j-1)$th mode for all $2\leq j\leq m$.

Suppose that the coefficient tensor $\cm{A}$ has tensor train (TT) ranks $R=(r_{1},r_{2},\ldots,r_{m+n-1})$, i.e. $r_i=\rank([\cm{A}]_i)$ for all $1\leq i\leq m+n-1$, and then the parameter space can be denoted by
\[
\bm{\Theta}_1(R)=\{\cm{A}\in\mathbb{R}^{q_1\times\cdots\times q_n\times p_1\times\cdots\times p_m}: \rank([\cm{A}]_i)\leq r_i \text{ for all } 1\leq i\leq m+n-1\}.
\]
From Proposition \ref{prop1}, we have $\cm{A}=[[\bm{G}_1,\cm{G}_2,\ldots,\cm{G}_n,\bm{\Sigma},\cm{G}_{n+1},\ldots,\cm{G}_{n+m-1},\bm{G}_{n+m}]]$, i.e.
\[
[\cm{A}]_n=(\bm{I}_{q_2\cdots q_n}\otimes\bm{G}_1)\cdots(\bm{I}_{q_n}\otimes[\cm{G}_{n-1}]_2)[\cm{G}_n]_2\bm{\Sigma}[\cm{G}_{n+1}]_1
([\cm{G}_{n+2}]_1\otimes\bm{I}_{p_{1}})\cdots(\bm{G}_{n+m}^\top\otimes\bm{I}_{p_{1}\cdots p_{m-1}}),
\]
where $\bm{G}_1^\top\bm{G}_1=\bm{I}_{r_1}$, $[\cm{G}_i]_2^\top[\cm{G}_i]_2=\bm{I}_{r_i}$ for $2\leq i\leq n$, $[\cm{G}_i]_1[\cm{G}_i]_1^\top=\bm{I}_{r_{i-1}}$ for $n+1\leq i\leq n+m-1$, $\bm{G}_{n+m}^\top\bm{G}_{n+m}=\bm{I}_{r_{n+m-1}}$, and $\bm{\Sigma}\in\mathbb{R}^{r_n\times r_n}$ is diagonal.
Note that $\vectorize(\cm{Y}_i) = [\cm{A}]_n\vectorize(\cm{X}_i)+\vectorize(\cm{E}_i)$, and then model \eqref{regmodel} can be rewritten into
\begin{align}
	\begin{split}\label{interpretation}
	[\cm{G}_n]_{2}^\top(\bm{I}_{q_n}\otimes[\cm{G}_{n-1}]_{2}^\top)&\cdots(\bm{I}_{q_2\cdots q_n}\otimes\bm{G}_{1}^\top)\vectorize(\cm{Y}_i) \\
	&=\bm{\Sigma} \{[\cm{G}_{n+1}]_1
	([\cm{G}_{n+2}]_1\otimes\bm{I}_{p_{1}})\cdots(\bm{G}_{n+m}^\top\otimes\bm{I}_{p_{1}\cdots p_{m-1}})\vectorize(\cm{X}_i)\}\\
	&\hspace{25mm}+[\cm{G}_n]_{2}^\top(\bm{I}_{q_n}\otimes[\cm{G}_{n-1}]_{2}^\top)\cdots(\bm{I}_{q_2\cdots q_n}\otimes\bm{G}_{1}^\top)\vectorize(\cm{E}_i).
	\end{split}
\end{align}
For simplicity, we call the above model \textit{tensor train regression}.

Similar to the illustrative example in Section \ref{adv:TT}, model \eqref{interpretation} consists of two stages: summarizing factors and conducting regression on the resulting factors in the element-wise sense.
In the first stage, we first perform dimension reduction for $\cm{Y}_i$ along the first mode by summarizing $q_1$ variables into $r_1$ factors through $\bm{G}_1$, then compress the stacked first and second modes by summarizing $r_1q_2$ factors into $r_2$ ones through $\cm{G}_2$, and so on.
The procedure will end until we extract $r_n$ factors from responses.	
On the contrary, the factor-extracting procedure for $\cm{X}_i$ follows a reverse order, and it rolls up from the last mode until the first mode one-by-one. At the end, it will also produce $r_n$ factors from predictors; see Section B.3 of the supplementary file for the illustration on how to extract factors from a third order predictor tensor.
Finally, the second stage regresses $r_n$ response factors on $r_n$ predictor ones with the coefficient matrix $\bm{\Sigma}$ being diagonal, and hence there are only $r_n$ parameters involved.

%

Consider the ordinary least squares (OLS) estimation for linear regression at \eqref{interpretation},
\begin{equation*}
	\cm{\widehat{A}}_{\mathrm{LR}}=\arg\min_{\cm{A}\in\bm{\Theta}_1(R)}\sum_{i=1}^{N}\|\cm{Y}_i-\langle \cm{A}, \cm{X}_i\rangle\|_{\mathrm{F}}^{2}.
\end{equation*}
Let $P=\prod_{i=1}^mp_i$ and $Q=\prod_{i=1}^nq_i$, and $d_{\mathrm{LR}}=\sum_{i=1}^mr_{n+i-1}p_ir_{n+i}+\sum_{i=1}^nr_{i-1}q_ir_i+r_n$ is the corresponding model complexity.
To establish nonasymptotic properties for the OLS estimator, we first make two assumptions on predictors $\bm{x}_i=\vectorize(\cm{X}_i)\in\mathbb{R}^{P}$ and error terms $\bm{e}_i =\vectorize(\cm{E}_{i}) \in\mathbb{R}^{Q}$, respectively, and they both are commonly used in the literature \citep{wainwright2019}.

\begin{assumption}\label{assump:input}
	Predictors $\{\bm{x}_i\}$ are $i.i.d$ with $\mathbb{E}(\bm{x}_i)=0$, $\mathrm{var}(\bm{x}_i)=\bm{\Sigma}_x$ and $\sigma^2$-sub-Gaussian distribution. There exist constants $0<c_{x}<C_{x}<\infty$ such that $c_x\bm{I}_P\leq\bm{\Sigma}_x\leq C_x\bm{I}_P$, where the two quantities depend on the dimensions of $p_i$'s, and they may shrink to zero or diverge to infinity as the dimensions increase.
\end{assumption}

\begin{assumption}\label{assump:error}
	Error terms $\{\bm{e}_i\}$ are $i.i.d$ and, conditional on $\bm{x}_i$, $\bm{e}_i$ has mean zero and follows  $\kappa^2$-sub-Gaussian distribution.
\end{assumption}

\begin{theorem}\label{thm:reg}
	Suppose that Assumptions \ref{assump:input} and \ref{assump:error} hold, and the true coefficient tensor $\cm{A}^*$ has TT ranks $R=(r_1,r_2,\ldots,r_{m+n-1})$.
	If the sample size $N\gtrsim\max \{\sigma^{2}/c_x,\sigma^{4}/c_x^2\}\cdot \{m\log(m\sqrt{C_x/c_x})$
	$(\sum_{i=2}^mr_{n+i-1}p_ir_{n+i}+(p_1+1)r_{n+1})+\sum_{i=1}^n\log(q_i)\}$, then
	\begin{align*}
		\|\cm{\widehat{A}}_{\mathrm{LR}}-\cm{A}^*\|_{\mathrm{F}}&\lesssim \frac{\kappa}{c_x}\sqrt{\frac{C_{x}(m+n)\log(m+n)d_{\mathrm{LR}}}{N}}
	\end{align*}
	with probability at least $1-\exp \{-c(\sum_{i=1}^mp_1+\sum_{i=1}^n\log q_i)\}$, where $c$ is some positive constant.
\end{theorem}

\begin{remark}\label{rem2}
From the above theorem, when $m$ and $n$ are fixed, and the parameters $c_x$, $C_x$, $\sigma^2$ and $\kappa^2$ are bounded away from zero and infinity, we then have  $\|\cm{\widehat{A}}_{\mathrm{LR}}-\cm{A}^*\|_{\mathrm{F}}=O_p(\sqrt{d_{\mathrm{LR}}/N})$.
Moreover, if Tucker decomposition is applied to $\cm{A}$ with Tucker ranks $(r_1,\ldots,r_{m+n})$, then the resulting OLS estimator will have the convergence rate of $\sqrt{d_{\mathrm{Tucker}}/N}$ with the model complexity of $d_{\mathrm{Tucker}}=\prod_{i=1}^{m+n}r_i+\sum_{i=1}^mp_ir_{n+i}+\sum_{i=1}^nq_ir_i$ \citep{Han2020,Wang2021}.
The term of $\prod_{i=1}^{m+n}r_i$ may result in a huge number even for lower ranks, and this actually hinders the application of Tucker decomposition on tensor regression with higher order response and predictor tensors.
Of course, all decomposition methods provide approximation only to real datasets and, when the sample size $N$ of a dataset is large enough, Tucker decomposition may lead to a better fit since it involves more parameters and hence more flexibility of the resulting model.
\end{remark}

\begin{remark}\label{rem1}
For the technical proof of Theorem \ref{thm:reg}, it is the key to derive the covering number for a tensor $\cm{A}$ with TT decomposition, which is of interest independently.
We adapt the proving technique in \cite{candes2011tight} for the derivation, and the TT decomposition at Proposition \ref{prop1} plays a key role in the process.
Specifically, we need to keep one matrix or tensor component for amplitude, while the other orthonormal matrices act as rotation; see Lemma D.3 in the supplementary file for details. 
In fact, this technique has been used to establish the covering number of a tensor with Tucker decomposition \citep{Wang2021}.
On the other hand, we may also want to count the covering number for a tensor with CP decomposition at \eqref{CP-format}, while it will lead to a form of Tucker decomposition after orthonormalizing factor matrices $\bm{U}_i$'s with $1\leq i\leq d$.
In other words, the proving technique in \cite{candes2011tight} will lead to the same covering number for a tensor with CP and Tucker decomposition.
\end{remark}

\subsection{Tensor train autoregression}\label{sec:TT-AR}

We first consider a simple autoregressive (AR) model for tensor-valued time series $\{\cm{Y}_t\}$,
\begin{equation}\label{automodel}
	\cm{Y}_t = \langle\cm{A},\cm{Y}_{t-1}\rangle+\cm{E}_t,
\end{equation}
where $\cm{Y}_t\in\mathbb{R}^{p_1\times p_2\times\cdots\times p_d}$, $\cm{A}\in\mathbb{R}^{p_1\times\cdots\times p_d\times p_1\times\cdots\times p_d}$ is the coefficient tensor, and $\cm{E}_t\in\mathbb{R}^{p_1\times p_2\times\cdots\times p_d}$ is the error term.
As in Section \ref{sec:tt}, we rearrange $\cm{Y}_t$ into a nested structure, and its $j$th mode is nested under each level of $(j+1)$th mode for all $1\leq j\leq d-1$.
Moreover, denote by $\cm{Y}_{t}^{\top}$ the rearrangement of $\cm{Y}_{t}$ such that the $1,2,\ldots,d$-th modes of $\cm{Y}_{t}$ correspond to the $d,\ldots,2,1$-th modes of $\cm{Y}_{t}^{\top}$, respectively. As a result, $\cm{Y}_{t-1}^\top$ also has a nested structure but with a reverse order, and hence model \eqref{automodel} can be reorganized into
\begin{equation*}
	\cm{Y}_t = \langle \mathcal{M}(\cm{A}),\cm{Y}_{t-1}^\top\rangle+\cm{E}_t,
\end{equation*}
where $\mathcal{M}(\cm{A})$ is the rearrangement of $\cm{A}$ according to $\cm{Y}_{t-1}^\top$.
Suppose that the reshaped coefficient tensor $\mathcal{M}(\cm{A})$ has low TT ranks $R=(r_1,r_2,\ldots,r_{2d-1})$, i.e. $r_i=\rank([\mathcal{M}(\cm{A})]_i)$ for all $1\leq i\leq 2d-1$, and the corresponding parameter space can then be denoted by
\[
\bm{\Theta}_2(R)=\{\cm{A}\in\mathbb{R}^{p_1\times\cdots\times p_d\times p_1\times\cdots\times p_d}: \rank([\mathcal{M}(\cm{A})]_i)\leq r_i \text{ for all } 1\leq i\leq 2d-1\}.
\]
For simplicity, we call model \eqref{automodel}, together with the parameter space of $\bm{\Theta}_2(R)$, \textit{tensor train autoregression} with order one, and its physical interpretation is similar to that of tensor regression.

For an observed time series $\{\cm{Y}_0,\cm{Y}_1,\ldots,\cm{Y}_N\}$ generated by model \eqref{automodel} with TT ranks $R=(r_1,r_2,\ldots,r_{2d-1})$, the OLS estimation can be defined as
\begin{equation}\label{AR1}
	\cm{\widehat{A}}_{\mathrm{AR}}=\arg\min_{\cm{A}\in \bm{\Theta}_2(R)}\sum_{t=1}^{N}\|\cm{Y}_{t}-\langle \cm{A}, \cm{Y}_{t-1}\rangle\|_{\mathrm{F}}^{2}.
\end{equation}

\begin{assumption}\label{assump:radius}
	Sequential matricization $[\cm{A}]_d$ has the spectral radius strictly less than one.
\end{assumption}

\begin{assumption}\label{assump:errorauto}
	For the error term, let $\bm{e}_t =\vectorize(\cm{E}_{t}) =\bm{\Sigma}_e^{1/2}\bm{\xi}_t$ and $P=\prod_{i=1}^dp_i$.
	Random vectors $\{\bm{\xi}_t\}$ are $i.i.d.$ with $\mathbb{E}(\bm{\xi}_t)=0$ and $\mathrm{var}(\bm{\xi}_t)=\bm{I}_P$, the entries $(\bm{\xi}_{tj})_{1\leq j\leq P}$ of $\bm{\xi}_t$ are mutually independent and $\kappa^2$-sub-Gaussian distributed, and there exist constants $0<c_{e}<C_{e}<\infty$ such that $c_{e}\bm{I}_P\leq\bm{\Sigma}_e\leq C_{e}\bm{I}_P$.
\end{assumption}

Assumption \ref{assump:radius} is necessary and sufficient for the existence of a unique strictly stationary solution to model \eqref{automodel}. The sub-Gaussian condition in Assumption \ref{assump:errorauto} is weaker than the commonly used Gaussian assumption in the literature \citep{Basu2015}, and it thanks to the established martingale-based concentration bound, which is a nontrivial task for high-dimensional time series and is of interest independently, in the technical proofs.

We next establish the estimation error bound of $\cm{\widehat{A}}_{\mathrm{AR}}$, which will rely on the temporal and cross-sectional dependence of $\{\cm{Y}_t\}$ \citep{Basu2015}.
To this end, we first define a matrix polynomial $\cm{A}(z):=\bm{I}_P-[\cm{A}]_dz$, where $z\in\mathbb{C}$ with $\mathbb{C}$ being the complex space, and its conjugate transpose is denoted by $\bar{\cm{A}}(z)$.
Let
\begin{equation*}
	\mu_{\mathrm{min}}(\cm{A})=\mathrm{min}_{|z|=1}\lambda_{\mathrm{min}}(\bar{\cm{A}}(z)\cm{A}(z)) \hspace{5mm}\text{and}\hspace{5mm} \mu_{\mathrm{max}}(\cm{A})=\mathrm{max}_{|z|=1}\lambda_{\mathrm{max}}(\bar{\cm{A}}(z)\cm{A}(z)).
\end{equation*}
We then denote $d_{\mathrm{AR}}=\sum_{i=1}^d(r_{i-1}p_ir_i+r_{2d-i}p_ir_{2d-i+1})+r_d$, $\kappa_U=C_{e}/\mu_{\mathrm{min}}(\cm{A}^*)$ and $\kappa_L=c_{e}/\mu_{\mathrm{max}}(\cm{A}^*)$, where $\cm{A}^*(z)$ is the polynomial with $\cm{A}=\cm{A}^*$, and $\cm{A}^*$ is the true coefficient tensor.
	
\begin{theorem}\label{thm:errorbound-TR}
		Suppose that Assumptions \ref{assump:radius} and \ref{assump:errorauto} hold, and the reshaped true coefficient tensor $\mathcal{M}(\cm{A}^*)$ has TT ranks $R=(r_1,r_2,\ldots,r_{2d-1})$.
		If the sample size
		 $N\gtrsim\max (\kappa^2\kappa_U/\kappa_L,\kappa^4\kappa_U^2/\kappa_L^2) \cdot d\log(d\sqrt{\kappa_U/\kappa_L})d_{\mathrm{AR}}$, then
		\begin{align*}
			 \|{\cm{\widehat{A}}_{\mathrm{AR}}-\cm{A}^*}\|_{\mathrm{F}}&\lesssim \frac{\kappa}{\kappa_L}\sqrt{\frac{C_{e}\kappa_Ud\log(d)d_{\mathrm{AR}}}{N}}
		\end{align*}
		with probability at least $1-\exp\{-c\sum_{i=1}^dp_i\}$, where $c$ is some positive constant.
\end{theorem}

From the above theorem, when $d$ is fixed, and the parameters $c_e$, $C_e$, $\kappa^2$, $\mu_{\mathrm{min}}(\cm{A}^*)$ and $\mu_{\mathrm{max}}(\cm{A}^*)$ are bounded away from zero and infinity, we then have $\|\cm{\widehat{A}}_{\mathrm{AR}}-\cm{A}^*\|_{\mathrm{F}}=O_{P}(\sqrt{d_{\mathrm{AR}}/N})$, where $d_{\mathrm{AR}}$ is the corresponding model complexity.

The AR model with lag one at model \eqref{automodel} may not be enough to fit the data with more complicated autocorrelation structures, and we next extend it to the case with a general order $p$,
\begin{equation}\label{TT-ar}
	\cm{Y}_t = \langle\cm{A}_1,\cm{Y}_{t-1}\rangle+\langle\cm{A}_2,\cm{Y}_{t-2}\rangle+\ldots+\langle\cm{A}_p,\cm{Y}_{t-p}\rangle+\cm{E}_t.
\end{equation}
We first stack $\cm{Y}_{t-1},\ldots,\cm{Y}_{t-p}$ into a new tensor $\cm{Y}_{t-1:t-p}\in\mathbb{R}^{p\times p_1\times\cdots\times p_d}$ with $(d+1)$ modes, and then $\cm{A}_{1},\ldots,\cm{A}_{p}$ are also stacked into another new tensor $\cm{A}_{1:p}\in\mathbb{R}^{ p_1\times\cdots\times p_d\times p\times p_1\times\cdots\times p_d }$ with $(2d+1)$ modes accordingly.
Note that, for $\cm{Y}_t$, its $j$th mode is nested under each level of $(j+1)$th mode for all $1\leq j\leq d-1$, and we here actually put the new mode of lags at the lowest place in the hierarchical structure.
It certainly can be at the highest place, i.e. all modes are nested under each lag, and the notations can be adjusted accordingly.
On the other hand, as in \cite{Wang2021}, such arrangement will make it convenient to further explore the possible low-rank structure along lags.

Denote by $\cm{Y}_{t-1:t-p}^{\top}$ the rearrangement of $\cm{Y}_{t-1:t-p}$ such that the $1,2,\ldots,(d+1)$-th modes of $\cm{Y}_{t-1:t-p}$ correspond to the $(d+1),\ldots,2,1$-th modes of $\cm{Y}_{t-1:t-p}^{\top}$, respectively. 
As a result, model \eqref{TT-ar} can be rewritten into
\begin{equation}\label{automodel2}
	\cm{Y}_t = \langle\cm{A}_{1:p},\cm{Y}_{t-1:t-p}\rangle+\cm{E}_t \hspace{5mm}\text{or}\hspace{5mm}
	\cm{Y}_t = \langle \mathcal{M}(\cm{A}_{1:p}),\cm{Y}_{t-1:t-p}^\top\rangle+\cm{E}_t,
\end{equation}
where $\mathcal{M}(\cm{A}_{1:p})$ is the rearrangement of $\cm{A}_{1:p}$ according to $\cm{Y}_{t-1:t-p}^\top$.
Suppose that the reshaped coefficient tensor $\mathcal{M}(\cm{A}_{1:p})$ has low TT ranks $R=(r_1,r_2,\ldots,r_{2d})$, i.e. $r_i=\rank([\mathcal{M}(\cm{A})]_i)$ for all $1\leq i\leq 2d$, and then the parameter space for model \eqref{automodel2} can be denoted by
\[
 \bm{\Theta}_3(R)=\{\cm{A}\in\mathbb{R}^{p_1\times\cdots\times p_d\times p\times p_1\times\cdots\times p_d}: \rank([\mathcal{M}(\cm{A})]_i)\leq r_i \text{ for all } 1\leq i\leq 2d\}.
\]
We refer the tensor train autoregression with the order of $p$ to model \eqref{TT-ar} or \eqref{automodel2} with the parameter space of $\bm{\Theta}_3(R)$. For a generated time series $\{\cm{Y}_{1-p},\ldots,\cm{Y}_0,\cm{Y}_1,\ldots,\cm{Y}_N\}$, the OLS estimation is defined as
\begin{equation}\label{ARp}
  	\cm{\widehat{A}}_{1:p}=\arg\min_{\cm{A}\in \bm{\Theta}_3(R)}\sum_{t=1}^{N}\|\cm{Y}_{t}-\langle \cm{A}_{1:p}, \cm{Y}_{t-1:t-p}\rangle\|_{\mathrm{F}}^{2}.
\end{equation}

The matrix polynomial for model \eqref{TT-ar} can be defined as $\cm{A}(z):=\bm{I}_P-[\cm{A}_1]_dz-\cdots-[\cm{A}_p]_dz^p$, where $z\in\mathbb{C}$, and similarly the notations of $\mu_{\mathrm{min}}(\cm{A})$ and $\mu_{\mathrm{max}}(\cm{A})$.
Let $d_{\mathrm{M}}=\sum_{i=1}^d(r_{i-1}p_ir_i+r_{2d-i}p_ir_{2d-i+1})+r_d+pr_{2d}$, $\kappa_U=C_{e}/\mu_{\mathrm{min}}(\cm{A}^*)$ and $\kappa_L=c_{e}/\mu_{\mathrm{max}}(\cm{A}^*)$, where $\cm{A}^*(z)$ is the polynomial with $\cm{A}_{1:p}=\cm{A}_{1:p}^*$, and $\cm{A}_{1:p}^*$ is the true coefficient tensor.

\begin{assumption}\label{assump:radiusplus}
	The determinant of $\cm{A}(z)$ is not equal to zero for all $|z|<1$.
\end{assumption}

\begin{theorem}\label{thm:errorbound-TR-lag}
	Suppose that Assumptions \ref{assump:errorauto} and \ref{assump:radiusplus} hold, and the reshaped true coefficient tensor $\mathcal{M}(\cm{A}_{1:p}^*)$ has TT ranks $R=(r_1,r_2,\ldots,r_{2d})$. If the sample size $N\gtrsim\max\left(\kappa^2\kappa_U/\kappa_L,\kappa^4\kappa_U^2/\kappa_L^2\right) \cdot d\log(d\sqrt{\kappa_U/\kappa_L})d_{\mathrm{M}}$, then
	\begin{align*}
		\|\cm{\widehat{A}}_{1:p}-\cm{A}_{1:p}^*\|_{\mathrm{F}}&\lesssim \frac{\kappa}{\kappa_L}\sqrt{\frac{C_{e}\kappa_Ud\log(d)d_{\mathrm{M}}}{N}}
	\end{align*}
	with probability at least $1-\exp\{-c(\sum_{i=1}^dp_i+p)\}$, where $c$ is some positive constant.
\end{theorem}

Note that $\|\cm{\widehat{A}}_{1:p}-\cm{A}_{1:p}^*\|_{\mathrm{F}}^2 =\sum_{j=1}^p\|\cm{\widehat{A}}_j-\cm{A}_j^*\|_{\mathrm{F}}^2$, where $\cm{\widehat{A}}_j$'s and $\cm{A}_j^*$'s are OLS estimators and true coefficients, respectively.
From Theorem \ref{thm:errorbound-TR-lag}, when $d$ is fixed, and the parameters $c_e$, $C_e$, $\kappa^2$, $\mu_{\mathrm{min}}(\cm{A}^*)$ and $\mu_{\mathrm{max}}(\cm{A}^*)$ are bounded away from zero and infinity, it holds that $\sum_{j=1}^p\|\cm{\widehat{A}}_j-\cm{A}_j^*\|_{\mathrm{F}}=O_{P}(\sqrt{d_{\mathrm{M}}/N})$, where $d_{\mathrm{M}}$ is the corresponding model complexity.

\section{Algorithm and its theoretical justifications}\label{Implementation}
\subsection{Riemannian gradient descent algorithm}

The three estimation problems in the previous section can be summarized into
\begin{equation}\label{optimizationRGrad}
	\cm{\widehat{A}}=\arg\min_{\cm{A}\in\bm{\Theta}} \mathcal{L}_N(\cm{A}) \hspace{3mm}\text{with}\hspace{3mm}
	\mathcal{L}_N(\cm{A})=\frac{1}{N}\sum_{i=1}^{N}\|\cm{Y}_i-\langle \cm{A}, \cm{X}_i\rangle\|_{\mathrm{F}}^{2},
\end{equation}
where $(\bm{\Theta},\cm{A},\cm{X}_i)=(\bm{\Theta}_1(R),\cm{A},\cm{X}_i)$ for regression, $(\bm{\Theta}_2(R),\mathcal{M}(\cm{A}),\cm{Y}_{i-1}^\top)$ for autoregression with order one, and $(\bm{\Theta}_3(R),\mathcal{M}(\cm{A}_{1:p}),\cm{Y}_{i-1:i-p}^\top)$ for autoregression with a general order $p$.
Suppose that TT ranks $R=(r_1,\ldots,r_{m+n-1})$ are known. The optimization at \eqref{optimizationRGrad} is nonconvex, and it usually is challenging to solve it both numerically and theoretically.

There are three types of commonly used algorithms in the literature for tensor-related optimization: projected gradient descent methods \citep{Chen2019,Han2020,hao2021sparse}, alternating minimization \citep{Zhou2013,Sun2019} and Riemannian optimization methods \citep{luo2022tensor, cai2022provable,kressner2016preconditioned}.
Despite its popularity, the projected gradient descent method may suffer from high computational cost since the vanilla gradient matrix is usually full-rank, and it is time-consuming to compute its singular value decomposition (SVD) at each iteration.
In the meanwhile, the alternating minimization methods were demonstrated numerically to be more sensitive to rank misspecification, especially when being compared with Riemannian optimization approaches \citep{luo2022tensor}.

This paper introduces a Riemannian gradient descent algorithm for the optimization problem at \eqref{optimizationRGrad}; see Algorithm \ref{alg:RGrad} for details.
Note that the tensors with fixed TT ranks, such as $\bm{\Theta}_1(R)$, $\bm{\Theta}_2(R)$ or $\bm{\Theta}_3(R)$, will form a Riemannian manifold \citep{holtz2012manifolds}, and hence the constrained optimization problem at \eqref{optimizationRGrad} can be recast to an unconstrained one over the Riemannian manifold.
As a result, a three-stage procedure can be conducted at each iteration: projecting the vanilla gradient matrix onto a tangent space of the manifold, performing the gradient descent along the tangent space, and finally retracting the resulting values back to the manifold.
Compared to projected gradient descent methods, the Riemannian gradient on the tangent space is low-rank, and this can speed up the retraction stage significantly.
Moreover, the Riemannian gradient can be efficiently calculated by a closed-form solution, making the algorithm  appealing \citep{Zhou2020}. 
The Riemannian optimization method has been applied to Tucker tensor regression \citep{luo2022tensor} and tensor train completion \citep{cai2022provable}, and its efficiency has been demonstrated for low-rank tensor-related optimization.

\begin{algorithm}
	\caption{Riemannian gradient descent for tensor train regression}\label{alg:RGrad}
	\begin{algorithmic}[H]
		\State Input: $\{(\cm{Y}_i,\cm{X}_i), 1\leq i\leq N\}$, TT ranks $R=(r_{1},\ldots,r_{m+n-1})$, number of iterations $K$, initial values $\cm{A}^{0}\in\bm{\Theta}(R)$.
		\For{$k=0,1,\ldots,K-1$}
		\State $\cm{A}^{k+0.5}=\cm{A}^{k}-\alpha_kP_{\mathbb{T}_{\scalebox{0.5}{\cm{A}}^k}}\nabla \mathcal{L}(\cm{A}^{k})$
		\State $\cm{A}^{k+1}=\textrm{TT-SVD}(\cm{A}^{k+0.5})$ 
		\EndFor
		\State \Return $\cm{A}^{K}.$
	\end{algorithmic}
\end{algorithm}

Consider all tensors $\cm{Z}\in\mathbb{R}^{q_1\times\cdots \times q_n}$ with TT ranks $(r_1,\ldots,r_{n-1})$, and then a Riemannian manifold is formed.
For a point $\cm{A}=[[\bm{G}_1,\cm{G}_2,\cdots,\cm{G}_{n-1},\bm{G}_n]]$ of the Riemannian manifold, we assume that the first $n-1$ TT cores are orthonormal, i.e. $\bm{G}_1^\top\bm{G}_{1}=\bm{I}_{r_1}$ and $[\cm{G}_i]_2^\top[\cm{G}_i]_2=\bm{I}_{r_i}$ with $2\le i\le n-1$, and this can be obtained by the TT-SVD \citep{oseledetsTR}.
Denote by $\mathbb{T}_{\scalebox{0.5}{\cm{A}}}$ the tangent space of the manifold at point
$ \cm{A}$ and then, for any tensor $\cm{Z}\in\mathbb{R}^{q_1\times\cdots\times q_n}$, we can give its projection onto the tangent space $\mathbb{T}_{\scalebox{0.5}{\cm{A}}}$ below,
\begin{equation}\label{RGradient}
	P_{\mathbb{T}_{\scalebox{0.5}{\cm{A}}}}\cm{Z}=\delta\cm{Z}_1+\cdots+\delta\cm{Z}_{n} \quad \text{with} \quad \delta\cm{Z}_i=[[\bm{G}_1,\cdots,\cm{G}_{i-1},\cm{Z}_i,\cm{G}_{i+1},\cdots,\bm{G}_{n}]],
\end{equation}
where $\cm{Z}_i$ is a tensor of $r_{i-1}\times q_i\times r_i$ for each $1\leq i\leq n$, and $r_0=r_n=1$; see \cite{holtz2012manifolds}.
Moreover, from the derivation of Section 4.3 in \cite{cai2022provable}, it holds that
\begin{align*}
	[\cm{Z}_i]_2 = (\bm{I}_{q_ir_{i-1}}-[\cm{G}_i]_2[\cm{G}_i]_2^\top)(\bm{I}_{q_i}\otimes\bm{G}^{\leq i-1})^\top[\cm{Z}]_i(\bm{G}^{\geq i+1})^\top(\bm{G}^{\geq i+1}(\bm{G}^{\geq i+1})^\top)^{-1}
\end{align*}
for $1\le i\le n-1$, and $\cm{Z}_n = [\cm{Z}]_{n-1}^\top\bm{G}^{\leq n-1}$, where $\bm{G}^{\leq i}$'s and $\bm{G}^{\geq i}$'s are defined in Section \ref{sec:2-1}.

We now introduce the three-stage procedure at each iteration.
At the first stage, the Riemannian gradient is obtained by projecting the vanilla gradient $\nabla \mathcal{L}(\cm{A}^{k})$ onto the tangent space $\mathbb{T}_{\scalebox{0.5}{\cm{A}}^k}$ at the current estimate $\cm{A}^{k}$. This projection is performed using the operator $P_{\mathbb{T}_{\scalebox{0.5}{\cm{A}}^k}}$, which projects any tensor onto the tangent space of $\cm{A}^k$. 
The second stage conducts gradient descent along the tangent space, and the step size $\alpha_k$ is chosen dynamically as
\begin{equation}\label{alg:step}
	\begin{split}
		\alpha_k := & \arg\min_{\alpha\in\mathbb{R}}\frac{1}{2}\sum_{i=1}^N\|\cm{Y}_j-\langle\cm{A}^k-\alpha P_{\mathbb{T}_{\scalebox{0.5}{\cm{A}}^k}}\nabla \mathcal{L}(\cm{A}^{k}),\cm{X}_i\rangle\|_{\mathrm{F}}^{2}\\
		=& \frac{\|P_{\mathbb{T}_{\scalebox{0.5}{\cm{A}}^k}}\nabla \mathcal{L}(\cm{A}^{k})\|_{\mathrm{F}}^{2}}{N^{-1}\sum_{i=1}^N\|\langle P_{\mathbb{T}_{\scalebox{0.5}{\cm{A}}^k}}\nabla \mathcal{L}(\cm{A}^{k}),\cm{X}_i\rangle\|_{\mathrm{F}}^{2}}.
	\end{split}
\end{equation} 
The last stage is to retract $\cm{A}^{k+0.5}=\cm{A}^{k}-\alpha_kP_{\mathbb{T}_{\scalebox{0.5}{\cm{A}}^k}}\nabla \mathcal{L}(\cm{A}^{k})$ back to the manifold by using the singular value decomposition for tensor train (TT-SVD) in \cite{oseledetsTR}. 
{ Finally, the proposed algorithm is a gradient descent method in nature, and it can be guaranteed to converge to at least a local stationary point \citep{chen2019nonconvex}.}

\subsection{Convergence analysis}

This subsection conducts convergence analysis for the Riemannian gradient descent algorithm, and its proving technique is similar to the Riemannian optimization methods for Tucker tensor regression \citep{luo2022tensor} and tensor train completion \citep{cai2022provable}.

\begin{assumption}\label{ass:iid entry}
	Each entry of predictors ${\cm{X}_i}$ is independently drawn from a $\sigma^2$-sub-Gaussian distribution with mean zero and variance one.
\end{assumption}

Note that the true coefficient tensor $\cm{A}^\ast$ has TT ranks $R=(r_1,\ldots,r_{m+n-1})$. Let $\underline{\lambda}:=\min_{1\leq k\leq m+n-1}\{\sigma_{r_k}([\cm{A}^\ast]_k)\}$, where $\sigma_{r}(\bm{A})$ denotes the $r$-th largest singular value of a matrix $\bm{A}$, and it refers to the minimum nonzero singular values of all sequential matricizations. Intuitively, it should not be too small to battle with noise. Moreover, three constants $\mu_{\ell}\in(0,1)$ with TT ranks $\ell=2R$, $4R$ and $5R$ come from the restricted isometry property in Proposition 2 of the supplementary file, and they are closely related to the shape of objective functions. 

\begin{theorem}\label{thm:alg}
	Suppose that Assumption \ref{assump:error} and \ref{ass:iid entry} hold, $\underline{\lambda}\ge c_2\kappa(m+n)^{5/2}\{\log(m+n)\}^{1/2}$ $(1+\mu_{5R}-\mu_{2R})({d_{\mathrm{LR}}}/{N})^{1/2}/(1-\mu_{2R})$, the initial value $\cm{A}^0\in\bm{\Theta}$ is chosen such that $\norm{\cm{A}^0-\cm{A}^\ast}_\mathrm{F}\leq c_1\underline{\lambda}(m+n)^{-3/2}/(1+\mu_{5R}-\mu_{2R})$, $\max(\mu_{2R},\mu_{4R})\leq{c_3}{(m+n)^{-1/2}}$, and the sample size $N\gtrsim\mathrm{max}(\sigma^2,\sigma^4)(m+n)\{(p_1r_{n+1}+\sum_{i=2}^mp_ir_{i+n-1}r_{i+n})\log m+\sum_{i=1}^n\log q_i\}$. After the $K$-th iteration of Algorithm \ref{alg:RGrad}, we have 
	\begin{equation}\label{eq:thmalg}
		\norm{\cm{A}^K-\cm{A}^\ast}_\mathrm{F}\leq 0.5^{K}\norm{\cm{A}^0-\cm{A}^\ast}_\mathrm{F}+\frac{c_4\kappa(m+n)\{\log(m+n)\}^{1/2}}{1-\mu_{2R}}\left(\frac{d_\mathrm{LR}}{N}\right)^{1/2},
	\end{equation}
	with probability at least $1-\exp\{-c_5(\sum_{i=1}^mp_i+\sum_{i=1}^n\log q_i)\}$, where $c_1$, $c_2$, $c_3$, $c_4$ and $c_5$ are some positive constants.
\end{theorem}

The two terms at the right-hand side of \eqref{eq:thmalg} correspond to the optimization and statistical errors, respectively. The statistical error has a form similar to that in Theorem \ref{thm:reg}. For any small $\epsilon>0$, we can choose the number of iterations $K=\{\log \|\cm{A}^{0}-\cm{A}^*\|_{\mathrm{F}}-\log\epsilon\}/\log2$ such that the optimization error has a value smaller than $\epsilon$.

\begin{remark} (Independent entries assumption).
Each entry of $\cm{X}$ is assumed to be independent in Theorem \ref{thm:alg}, and this is mainly for accommodating the dynamically chosen step size $\alpha_k$ since it requires the stronger condition of restricted isometry property instead of the restricted strong convexity. 
This condition is also needed to establish convergence analysis for Riemannian gradient descent algorithms in \cite{luo2022tensor} for Tucker tensor regression and in \cite{cai2022provable} for tensor train completion. 
For predictors with correlated entries, we derive the corresponding linear convergence theoretically in Theorem D.1 of the supplementary file, while the step size $\alpha_k$ at \eqref{alg:step} cannot be used anymore.
The theoretical justification for tensor autoregression is more challenging although the algorithm works well in practice, and we leave it for future research.
The step size $\alpha_k$ is multiplied by 0.1 to stabilize the optimization for tensor autoregression in all simulation studies and real analysis in this paper.
\end{remark}

\begin{remark} (Alternative method for retraction).
From the technical proof of Theorem \ref{thm:alg}, the approximation error caused by TT-SVD at the retraction stage has a non-negligible effect on the convergence.
It hence is of interest to alternatively consider the one-step tensor train orthogonal iteration (TTOI) \citep{Zhou2020} since it usually leads to a smaller approximation error, and accordingly different assumptions will be required. We leave it for future research.
\end{remark}

\begin{remark} (Initialization of the algorithm).
For scalar-to-tensor and tensor-to-vector regression models, by following \cite{luo2022tensor}, we consider the spectral method and a version based on QR decomposition for the initialization, respectively, and their theoretical justifications are also given in the supplementary file.
In the meanwhile, for a general tensor regression, it is usually difficult to provide a theoretically justified initialization, and this paper simply employs the spectral method, i.e. we conduct the one-step TTOI to  $N^{-1}\sum_{i=1}^N\cm{Y}_i\circ\cm{X}_i$ and use the resultant as an initialization.  
Finally, for tensor autoregression, the task is much more challenging, and the spectral method also does not work well numerically.
We suggest to run a few iterations of an approximate projected gradient descent algorithm with $\cm{A}^0$ being zero tensor.
More details on initialization can be found in Section C of the supplementary file.
\end{remark}

\subsection{Rank selection}\label{rankselection}

TT ranks are assumed to be known in the proposed methodology, while they are unknown in most real applications.
This subsection introduces an information criterion to select them.

Denote by $\cm{\widehat{A}}(R)$ the estimator from \eqref{optimizationRGrad}, where TT ranks $R=(r_1,\ldots,r_{m+n-1})$, and then the Bayesian information criterion (BIC) can be constructed below,
\begin{equation}\label{eq:IC}
	\mathrm{BIC}(R)=N\log\left\{N^{-1}\sum_{i=1}^{N}\|\cm{Y}_{i}-\langle \cm{\widehat{A}}(R), \cm{X}_{i}\rangle\|_{\mathrm{F}}^{2}\right\}+\phi d(R)\log(N) ,
\end{equation}
where $\phi$ is a {tuning} parameter, and $d(R)=\sum_{i=1}^nq_ir_ir_{i-1}+\sum_{j=1}^mp_jr_{n+j}r_{n+j-1}+r_n$ with $r_0=r_{m+n}=1$ is the number of free parameters.
As a result, TT ranks can be selected by $\widehat{R}_{\mathrm{Joint}}=\argmin_{r_i\leq \bar{r},1\leq i\leq m+n-1}\mathrm{BIC}(R)$, where $\bar{r}$ is a predetermined upper bound of ranks, while the computational cost will be very high for larger values of $m$ and $n$.
We may search for the best rank separately for each mode. Specifically, define
\[
\widehat{r}_{j}=\argmin_{r_{j}\le \bar{r}}\mathrm{BIC}(\bar{r},\ldots,\bar{r},r_{j},\bar{r},\ldots,\bar{r})
\]
for $1\leq j\leq m+n-1$, and the selected TT ranks are then $\widehat{R}_{\mathrm{Separate}}=(\widehat{r}_1,\ldots,\widehat{r}_{m+n-1})$.

We next provide theoretical justifications for the above selecting criteria. 
Assume that the true TT ranks are $R=(r_1,\ldots,r_{m+n-1})$, and let $\underline{\lambda} =\min_{1\le k\le m+n-1}\{\sigma_{r_k}([\cm{A}^\ast]_k)\}$ be the minimum of nonzero singular values of all sequential matricizations, where $\cm{A}^\ast$ represents the true coefficient tensor $\cm{A}^\ast$ for regression, $\mathcal{M}(\cm{A}^\ast)$ for autoregression with order one, or $\mathcal{M}(\cm{A}_{1:p}^\ast)$ for autoregression with a general order $p$, respectively.
Moreover, denote $\bar{d}=\bar{r}^2(\sum_{i=1}^nq_i+\sum_{i=1}^{m}p_i)+\bar{r}$. The consistency for rank selection is then given below. 


\begin{theorem}\label{thm:rank}
	Suppose that $\underline{\lambda}^2\geq
	\kappa^2C_{x}(m+n)\log(m+n)\bar{d}/(c_x^2N)+c_1$,  or $\kappa^2\kappa_UC_{e}d\log(d)\bar{d}/(\kappa_L^2N)+c_1$ for tensor train regression or autoregression, respectively, where $c_1$ is some positive constant given in the proof, and the other notations are defined for each of the three models in Section 3.
	If the conditions of Theorem \ref{thm:reg}, \ref{thm:errorbound-TR} or \ref{thm:errorbound-TR-lag} hold, and $r_j\leq \bar{r}$ for all $1\leq j\leq m+n-1$, then $\mathbb{P}\{\widehat{R}=R\}\rightarrow 1$ for both $\widehat{R}=\widehat{R}_{\mathrm{Joint}}$ and $\widehat{R}_{\mathrm{Separate}}$ as $N\rightarrow\infty$.
\end{theorem}

\begin{remark}
The above theorem provides an asymptotic result, while all the other theorems in this paper give nonasymptotic results. 
Consider tensor autoregression with order one and, theoretically speaking, we can select the right ranks with probability approaching one as $N\rightarrow\infty$, as long as $ d\log(d)\overline{d}/(\Pi_{i=1}^{d}p_i\log N)\ll \phi\ll N\log\{1+(\Pi_{i=1}^{d}p_i)^{-1}\}/\log N$. For low-dimensional setups, the range for feasible $\phi$ becomes larger when the sample size $N$ increases, and it contains $1/2$ for a sufficiently large $N$, i.e. it will reduce to the classical BIC. However, for high-dimensional setups, the feasible values of $\phi$ will depend on the whole parameter space, and we set it to 0.02 in all simulations and empirical examples with good performance.
\end{remark}

\section{Simulation studies}\label{sim}

This section conducts two simulation experiments to evaluate the finite-sample performance of OLS estimators for TT regression and to compare CP, Tucker and TT decomposition for tensor regression, respectively.


The data generating process in the first experiment is TT regression model at \eqref{regmodel} with $(n,m)=(2,3)$, i.e. $\cm{Y}_i\in\mathbb{R}^{q_1\times q_2}$ and $\cm{X}_i\in\mathbb{R}^{p_1\times p_2\times p_3}$.
For the components of predictors $\cm{X}_i$ and errors $\cm{E}_i$, they are (i.) independent with uniform distribution on $(-0.5,0.5)$, (ii.) independent with a standard normal distribution, or (iii.) correlated with normality, i.e. $\vectorize(\cm{X}_i){\sim}N(\bm{0},\bm{\Sigma}_x)$ and $\vectorize(\cm{E}_i) {\sim}N(\bm{0},\bm{\Sigma}_e)$ with $\bm{\Sigma}_x=(0.5^{|i-j|})_{1\leq i,j\leq P}$ and  $\bm{\Sigma}_e=(0.5^{|i-j|})_{1\leq i,j\leq Q}$, where $P=p_1p_2p_3$ and $Q=q_1q_2$.
The coefficient tensor has a TT decomposition $\cm{A}=[[\bm{G}_1,\cm{G}_2,\bm{\Sigma},\cm{G}_{3},\cm{G}_{4},\bm{G}_5]]\in\mathbb{R}^{q_1\times q_2\times p_1\times p_2\times p_3}$, and TT ranks are set to $R=(r,r,r,r)$ for simplicity. Orthonormal matrices
$\bm{G}_1,[\cm{G}_2]_2,[\cm{G}_3]_1,[\cm{G}_4]_1$ and $\bm{G}_5$ are generated by extracting leading singular vectors of Gaussian random matrices, and
$\bm{\Sigma}$ is a Gaussian random diagonal matrix and rescaled such that $\norm{\bm{\Sigma}}_{\mathrm{F}}=5$.

From Theorem \ref{thm:reg}, the estimation error has the convergence rate of $\sqrt{d_{\mathrm{LR}}/N}$ with 
$d_{\mathrm{LR}}=r^2(p_1+p_2+q_2)+r(p_3+q_1) +r$. To numerically verify the convergence rate, we consider four settings:
(a) $(p_i,q_j,r)$ is fixed at $(5,4,2)$ for all $i\in\{1,2,3\}$ and $j\in\{1,2\}$, while the sample size $N$ varies among the set of $\{200, 300, 400, 700, 1000, 1500\}$;
(b) $(p_i,q_j,N)$ is fixed at $(6,6,1000)$ for all $i\in\{1,2,3\}$ and $j\in\{1,2\}$, while the rank $r$ varies from one to six;
(c) $(r,N)$ is fixed at $(1,600)$, while we use the same value for all $p_i$'s and $q_j$'s, and it varies among $\{4, 5, 6, 7, 8, 10\}$; and
(d) $(r,N,r^2[p_1+p_2+q_2]+r[p_3+q_1])$ is fixed at $(1,600,26)$, while we vary $(p_1,p_2,p_3,q_1,q_2)$ among six different values, given in Section E.1 of the supplementary file.
Algorithm \ref{alg:RGrad} is employed to search for OLS estimators $\cm{\widehat{A}}_{\mathrm{LR}}$ with the step size dynamically chosen as in \eqref{alg:step}. We use spectral initialization mentioned in Section 4.2.
Estimation errors $\|\cm{\widehat{A}}_{\mathrm{LR}}-\cm{A}\|_{\mathrm{F}}$, averaged over 300 replications, are calculated and given in Figure \ref{sim:regfig}.

The linearity in Figure \ref{sim:regfig}(a)-(c) implies that $\|\cm{\widehat{A}}_{\mathrm{LR}}-\cm{A}\|_{\mathrm{F}}$ is proportional to $1/\sqrt{N}$ and $\sqrt{d_\mathrm{LR}}$. 
Moreover, the constancy in Figure \ref{sim:regfig}(d) further verifies that the estimation error depends on $p_i$'s and $q_j$'s linearly, which is one of the main advantages of tensor train regression.
The theoretical findings in Theorem \ref{thm:reg} are hence confirmed.
Moreover, for the correlated data from (iii.), all estimators have a slightly worse performance, and a generalized least squares method may be needed to improve their efficiency.
The finite-sample performance of OLS estimators for TT autoregression in handling time series data was also evaluated by simulation experiments, and similar findings can be observed; see Section E.3 in the supplementary file for details.

In the second experiment, the data are generated from model \eqref{regmodel} with $n=1$ and $2\leq m\leq 5$, i.e. $\bm{y}_i\in\mathbb{R}^{q_1}$ and $\cm{X}_i\in\mathbb{R}^{p_1\times \cdots\times p_m}$, and the components of predictors $\cm{X}_i$ and errors $\cm{E}_i$ are independently sampled from a standard normal distribution.
We set $q_1=p_1=p_2=8$ and $p_3=\cdots=p_m=5$, and the sample size is fixed at $N=500$.
To set the coefficient tensor $\cm{A}\in\mathbb{R}^{q_1\times p_1\times \cdots\times p_m}$, we consider three different decomposition methods: CP decomposition $(\cm{A}_\mathrm{CP})$ in \eqref{CP-format} with the CP rank  $r^\mathrm{CP}$, Tucker decomposition $(\cm{A}_\mathrm{Tu})$ in \eqref{Tucker-format} with Tucker ranks $r^\mathrm{Tu}_i$'s for $1 \leq i \leq m+1$, and TT decomposition $(\cm{A}_\mathrm{TT})$ with TT ranks  $r^\mathrm{TT}_i$'s for $0\leq i \leq m+1$ and $r^\mathrm{TT}_0=r^\mathrm{TT}_{m+1}=1$.
Note that a tensor with one of the three decomposition methods can be rewritten into that with any of the other two, and the relationships among these ranks are given below, 
\begin{align}
	\begin{split}\label{simu}
	r_i^{\mathrm{TT}}\leq r^\mathrm{CP}, \quad 
	r_i^\mathrm{Tu}\leq r^\mathrm{CP}, \quad 
	r^\mathrm{CP}\leq r_1^{\mathrm{TT}} \cdots  r_m^{\mathrm{TT}}, \quad
	r_i^{\mathrm{Tu}}\leq r_{i-1}^{\mathrm{TT}}r_i^{\mathrm{TT}}, \\
	r^\mathrm{CP}\leq r_1^{\mathrm{Tu}} \cdots  r_m^{\mathrm{Tu}}, \quad\text{and}\quad
	r_i^{\mathrm{TT}}\leq\mathrm{min}(r_1^{\mathrm{Tu}}\cdots r_i^{\mathrm{Tu}}, r_{i+1}^{\mathrm{Tu}}\cdots r_{m+1}^{\mathrm{Tu}});
	\end{split}
\end{align}
see \cite{cai2022provable} for details.
We set $r^\mathrm{CP}=r_1^{\mathrm{TT}}=\cdots=r_m^{\mathrm{TT}}=r_1^{\mathrm{Tu}}=\cdots=r_{m+1}^{\mathrm{Tu}}=3$ or 4, and $\norm{\cm{A}_\mathrm{CP}}_{\mathrm{F}}=\norm{\cm{A}_\mathrm{TT}}_{\mathrm{F}}=\norm{\cm{A}_\mathrm{Tu}}_{\mathrm{F}}=5$.
The generation of $\cm{A}_\mathrm{TT}$ and the orthonormal matrices, $\bm{U}_i$ with $1\leq i \leq m+1$,  in $\cm{A}_\mathrm{Tu}$ follow the procedures described in the first experiment; while the core tensor $\cm{G}$ in $\cm{A}_\mathrm{Tu}$ has independent entries from the standard normal distribution.

For the data generating process with $\cm{A}_\mathrm{CP}$, we conduct OLS estimation for tensor regression with three decomposition methods, and the corresponding estimators are denoted by $\cm{\widehat{A}}_{\mathrm{CP}}^\mathrm{CP}$, $\cm{\widehat{A}}_{\mathrm{CP}}^\mathrm{Tu}$ and $\cm{\widehat{A}}_{\mathrm{CP}}^\mathrm{TT}$, respectively, where the superscript indicates the type of estimators.
We utilize the projected gradient descent with the CP-ALS algorithm \citep{Kolda2009} and a step size of $\eta=0.01$ to search for the estimator with CP decomposition, $\cm{\widehat{A}}_{\mathrm{CP}}^\mathrm{CP}$.
Moreover, Algorithm \ref{alg:RGrad} and the Riemannian gradient algorithm of Tucker decomposition \citep{luo2022tensor} are adopted to search for $\cm{\widehat{A}}_{\mathrm{CP}}^\mathrm{TT}$ and $\cm{\widehat{A}}_{\mathrm{CP}}^\mathrm{Tu}$, respectively, and the ranks are set according to the relationships at \eqref{simu}. 
Similarly, we can conduct OLS estimation for data generating processes with $\cm{A}_\mathrm{Tu}$ and $\cm{A}_\mathrm{TT}$, respectively, and the notations can also be defined accordingly.
Figure \ref{sim:TTvsTucker} presents the averaged estimation errors over 300 replications, where $\|\cm{\widehat{A}}_\mathrm{CP}^{(\cdot)}-\cm{A}_\mathrm{CP}\|_{\mathrm{F}}$, $\|\cm{\widehat{A}}_\mathrm{TT}^{(\cdot)}-\cm{A}_\mathrm{TT}\|_{\mathrm{F}}$,
and $\|\cm{\widehat{A}}_\mathrm{Tu}^{(\cdot)}-\cm{A}_\mathrm{Tu}\|_{\mathrm{F}}$ are depicted in the left, middle and right panels, respectively. 

From the left panel of Figure \ref{sim:TTvsTucker}, since the true coefficient tensor has CP decomposition, $\cm{\widehat{A}}_\mathrm{CP}^\mathrm{CP}$ performs best as expected. The TT's estimation error $\norm{\cm{\widehat{A}}_\mathrm{CP}^\mathrm{TT}-\cm{A}_\mathrm{CP}}_\mathrm{F}$ exhibits a linear trend while the Tucker's estimation error $\norm{\cm{\widehat{A}}_\mathrm{CP}^\mathrm{Tu}-\cm{A}_\mathrm{CP}}_\mathrm{F}$ shows an exponential increase as $m$ increases. The difference becomes more pronounced when the true CP rank slightly changes from three to four.
Furthermore, comparing the middle and right panels, we observe that as $m$ and $r$ increases, when the true coefficient tensor has TT decomposition, Tucker estimator $\cm{\widehat{A}}_\mathrm{TT}^\mathrm{Tu}$ underperforms $\cm{\widehat{A}}_\mathrm{TT}^\mathrm{TT}$ more than that of TT estimator $\cm{\widehat{A}}_\mathrm{Tu}^\mathrm{TT}$ underperforms $\cm{\widehat{A}}_\mathrm{Tu}^\mathrm{Tu}$ when the true coefficient tensor has Tucker decomposition. 
We may conclude that Tucker decomposition is only suitable for low-order tensors, whereas TT decomposition is preferable for high order tensors with large $m$ and $r$.
Additionally, it is noteworthy to point out that the projected gradient descent algorithm for $\cm{\widehat{A}}_\mathrm{TT}^\mathrm{CP}$ and $\cm{\widehat{A}}_\mathrm{Tu}^\mathrm{CP}$ encounters non-convergence traps under the model misspecification due to numerical instability, and hence its results are absent from the middle and right panels. 

\section{An empirical example}\label{realdata}

Electrocorticography (ECoG) has been widely used to record brain activities for brain-machine interface (BMI) technology, and it has two categories: subdural and epidural ECoGs. The epidural ECoG is a more practical interface for real applications since it has long-term durability, and this subsection analyzes the epidural ECoG data in \cite{Shimoda2012}.

In detail, two adult Japanese monkeys were first implanted with 64-channel ECoG electrodes in the epidural space of the left hemisphere in Figure \ref{brainfig3}(a), and they were then trained to retrieve foods using the right hand.
Six markers were placed at the left and right shoulder, elbow and wrist, respectively, and their 3D positions were recorded at a certain moment $i$, where the three values of a position represent the left-right, forward-backward and up-down movements, respectively.
As a result, 3D positions are nested under each marker, and the hand position can be described by a matrix $\bm{Y}_i\in\mathbb{R}^{3\times6}$.
On the other hand, the ECoG signals from 64 channels were recorded with a sampling rate of 1 kHz per channel, and they were then preprocessed by the Morlet wavelet transformation at five center frequencies, 160, 80, 40, 30 and 20 Hz, where the cutoff frequencies are at 0.1 and 400 Hz; see \cite{zhao2012higher}.
Moreover, there are ten timestamps per second, and the ECoG signals in the previous second are believed to be able to affect hand positions.
Hence, a tensor is formed for predictors $\cm{X}_i\in\mathbb{R}^{64\times10\times5}$, where the three modes represent channel, timestamp and center frequency, respectively.
It has a factorial structure, and the arrangement will be explained later.
We use the records of 1045 seconds on August 2, 2010, and the hand position was extracted every 0.5 seconds. This results in $N=2089$ observations for a regression problem $\{(\bm{Y}_i,\cm{X}_i), 1\leq i\leq N\}$.

Tensor regression model at \eqref{regmodel} is applied with the coefficient tensor $\cm{A}\in\mathbb{R}^{3\times 6\times 64\times10\times5}$, and it has $57600$ parameters in total.
To conduct a feasible estimation, we have to restrict $\cm{A}$ onto a low-dimensional space, and four dimension reduction techniques are considered. They include TT, CP and Tucker decomposition, and the fourth one is an early stopping (ES) technique in the area of machine learning \citep{Goodfellow2016}, where we use the gradient descent method to update $\cm{A}^{(k)}$, and the total number of iterations are restricted to be the same as that of tensor train regression.
We first fit the TT regression to the whole dataset, and the selected TT ranks are $R_{\mathrm{Joint}}=(2, 1, 4, 1)$ with $\bar{r}=5$.
For CP decomposition at \eqref{CP-format}, the OLS estimation is considered as in \cite{lock2018tensor}, while the ridge regularization is removed. Moreover, the BIC at \eqref{eq:IC} is adapted for rank selection with $d(r)=r(\sum_{i=1}^3p_i+\sum_{i=1}^2q_i)$, $\phi=0.3$ and the upper bound of ranks $\bar{r}=15$; and the selected rank is $\widehat{r}=4$.
For Tucker decomposition at \eqref{Tucker-format}, the OLS estimation is employed again, and the selected Tucker ranks are $\widehat{R}=(2,2,1,2,2)$ by the 95\% cumulative percentage of total variation \citep{Han2020}.

We next evaluate the performance of these four dimension reduction techniques. We consider two Tucker regression methods including Riemannian Gradient descent (RGradTucker) in \cite{luo2022tensor} and projected gradient descent (PGDTucker) in \cite{Han2020}. For the sake of comparison, observations are split into two sets: 1500 samples for training and the remaining 589 for testing. We randomly shuffle the data 300 times and the boxplots of their $\ell_1$, $\ell_2$ and $\ell_{\infty}$ norms of forecast errors are displayed in Figure \ref{ECOG:fig}. Specifically, the $\ell_2$ norm of forecast errors is defined as $589^{-1}\sum_{i=1}^{589}\|\vectorize(\bm{Y}_i)-\vectorize(\bm{\widehat{Y}}_i) \|_2$, where $\bm{\widehat{Y}}_i =\langle\cm{\widehat{A}}, \cm{X}_i\rangle$ is the predicted value, and we can similarly define $\ell_1$ and $\ell_{\infty}$ norms of forecast errors.
From Figure \ref{ECOG:fig}, TT decomposition outperforms Tucker decomposition significantly, while it is also better than CP decomposition.
It may be due to the facts that, for high order tensors, Tucker decomposition fails to compress the space dramatically, while TT decomposition can reduce the parameter space as significantly as CP decomposition. In addition, CP decomposition exhibits a larger deviation due to its numerical instability.
The comparison between RGradTucker and PGDTucker also shows that the Riemannian gradient descent algorithm can restrict $\cm{A}$ onto a low-dimensional space more efficiently in contrast to the projected gradient descent.

Finally, we refit TT regression by using all observations with TT ranks $R=(2, 1, 4, 1)$ selected in the above, and Algorithm S.1 in the supplementary file is then employed to calculate TT decomposition, $\cm{\widehat{A}}= [[\bm{G}_1, \cm{G}_2, \bm{\Sigma}, \cm{G}_3, \cm{G}_4, \bm{G}_5]]$.
As a result, there are only one response factor $[\cm{G}_2]_{2}^\top\vectorize(\bm{G}_1^{\top}\bm{Y}_i)\in\mathbb{R}$ and one predictor factor $[\cm{G}_{3}]_1 ([\cm{G}_{4}]_1\otimes\bm{I}_{64}) (\bm{G}_{5}^\top\otimes\bm{I}_{640})
\vectorize(\cm{X}_i)\in\mathbb{R}$ in total; see Section 2.2 and \eqref{interpretation} for details.
For predictors, the orthonormal matrix $\bm{G}_5$ is first used to extract one factor from five center frequencies, which is the high-frequency band (or high-gamma band).
We then use $\cm{G}_4$ to summarize all 10 factors across timestamps into four ones, one of which is obviously for the high-gamma band, and finally the predictor factor can be obtained by applying $\cm{G}_3$; see loading matrices in Section F.1 of the supplementary file.
To further understand how the predictor factor summarizes center frequencies and timestamps, we average the absolute loading tensor along channels, and the corresponding heatmap is presented in Figure \ref{brainfig3}(b).
It can be seen that the loadings concentrate on the highest center frequency, while there is no clear trend for timestamps.
This is consistent with the well-known fact for ECoG that the motor behavior is highly related to the high-gamma band \citep{Schalk2007,Scherer2009}.
We then focus on the high-gamma band factor from $\bm{G}_5$ and $\cm{G}_4$, and the absolute loadings of its 64 channels from $\cm{G}_3$ are depicted in Figure \ref{brainfig3}(c) with heatmaps, where the black bounding box corresponds to the primary motor cortex.
Due to the importance of the primary motor cortex in decoding ECoG \citep{Shimoda2012}, the heavy loadings on electrode channels in the black bounding box further confirm the usefulness of the proposed TT regression.
For the arrangement of three modes in predictors, we place the channel at the highest level since it is the most important to detect influential electrode channels. In the meanwhile, from the above analysis, the high-frequency band plays a key role, and hence we place center frequencies at the lowest level such that they can be first summarized.

For the response factor, Figure \ref{brainfig3}(d) displays its heatmap of the absolute loading matrix, and six markers have roughly equal contributions to motor behavior. The forward-backward movement dominates the interpretation of hand positions, while the left-right or up-down movements have a different receptive field.
It is noteworthy to point out that the fitted tensor regression via Tucker decomposition has four response and four predictor factors since the Tucker ranks are $R=(2,2,1,2,2)$, while that via CP decomposition has predictor factors only.

\section{Conclusions and discussions}\label{sec:7}

This paper proposes a new tensor regression, tensor train (TT) regression, by revising the classical TT decomposition and then applying it to the coefficient tensor, and the new model is further extended to TT autoregression in handling time series data.
Comparing with the models via Tucker decomposition, the two new models can be applied to the case with higher order responses and predictors, while they have the same stable numerical performance.
Moreover, the proposed new models are shown to well match the data with hierarchical structures, and they also have better interpretation even for factorial data, which are supposed to be better fitted by Tucker decomposition.

The proposed methodology in this paper can be extended along three directions.
First, for TT regression, the sample size is required to be $N\gtrsim\sum_{i=1}^mp_i+\sum_{i=1}^nq_i$ such that the estimation consistency can be obtained, while the number of variables for some modes may be larger than the sample size $N$ in some real applications.
To handle this case, we may further impose sparsity to the components of TT decomposition \citep{Basu2015,Wang2021}.
Specifically, for the coefficient tensor $\cm{A}=[[\bm{G}_1,\cm{G}_2,\ldots,\cm{G}_n,\bm{\Sigma},\cm{G}_{n+1},\ldots,\cm{G}_{m+n-1},\bm{G}_{m+n}]]$ with TT ranks $(r_1,\ldots,r_{m+n-1})$, let $s_1$ and $s_{m+n}$ be the number of non-zero rows of $\bm{G}_1\in\mathbb{R}^{q_1\times r_1}$ and $\bm{G}_{m+n}\in\mathbb{R}^{q_{m+n}\times r_{m+n-1}}$, respectively, and $s_k$ be the number of non-zero lateral slices of $\cm{G}_k\in\mathbb{R}^{r_{k-1}\times q_k\times r_k}$ with $2\leq k\leq m+n-1$, where $q_{n+j}=p_j$ for $1\leq j\leq m$.
Note that $1\leq s_k\leq q_k$ for $1\leq k\leq m+n$, and we can conduct estimation by a hard thresholding method. 
Denote by $\cm{\widehat{A}}_{\mathrm{SP}}$ the corresponding estimator, and it can be verified that
$\|\cm{\widehat{A}}_{\mathrm{SP}}-\cm{A}^*\|_{\mathrm{F}}\lesssim \sqrt{s_{\max}(r_{\max}^2+\sum_{i=1}^m\log p_i+\sum_{i=1}^{n}\log q_i)/N}$ with high probability, where $s_{\max}=\mathrm{max}\{s_1, s_2, \cdots, s_{m+n}\}$, and the required sample size is significantly reduced.
The algorithms in Section \ref{Implementation} can also be adopted to search for estimators.
Secondly, the currently used factor modeling methods for tensor-valued time series are all based on Tucker decomposition \citep{Chen2021,wang2019factor} and, as mentioned in Section \ref{adv:TT}, they cannot be used for the data with hierarchical structures. It is of interest to use TT decomposition for factor modeling. 
{Finally, we may extend the proposed methodology to a generalized linear model with a scalar response $y$ and a tensor-valued predictor $\cm{X} \in \mathbb{R}^{p_1 \times \dots \times p_m}$ \citep{Zhou2013,chen2019nonconvex,tang2020individualized}. Specifically, suppose that the response is from the exponential family, and $g(\mathbb{E}(y|\cm{X})) = \langle \cm{A}, \cm{X} \rangle$ with $g(\cdot)$ being a suitable link function. 
We then apply TT decomposition to the coefficient tensor $\cm{A} \in \mathbb{R}^{p_1 \times \dots \times p_m}$, and the TT loading matrices will extract low-dimensional factors from predictors sequentially as in Section 3.}



\section*{Acknowledgement}
We are deeply grateful to the editor, the associate editor and three anonymous referees for their valuable comments that led to the substantial improvement of the manuscript. 
	
	\bibliography{TensorTrain}
	

\newpage

\begin{figure}[hb]
	\centering
	\includegraphics[width=0.8\textwidth,height=0.4\textwidth]{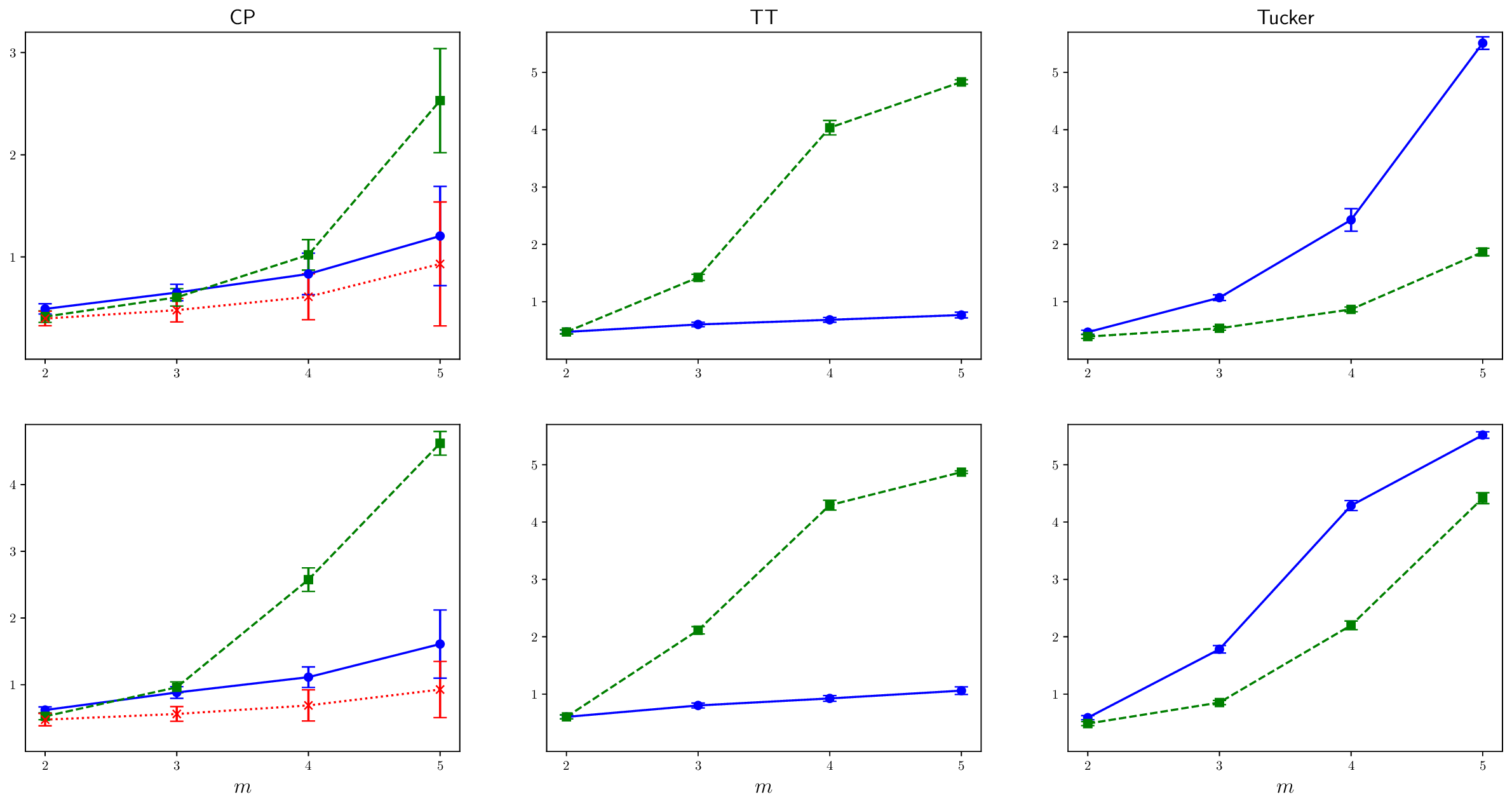}
	\caption{Averaged estimation errors $\|\cm{\widehat{A}}-\cm{A}\|_{\mathrm{F}}$ for tensor regression via Tucker (\protect\Lratetwenty), CP (\protect\Lratefive) and TT (\protect\Lrateten)  decomposition with true coefficient tensor $\cm{A}$ from CP (left panel), TT (middle panel), Tucker (right panel) decomposition and their corresponding ranks being $r=3$ (upper row) or 4 (lower row). One standard deviation is plotted above and below the averaged value.}
	\label{sim:TTvsTucker}
\end{figure}


\begin{figure}[ht]
	\centering
	\includegraphics[width=1\textwidth]{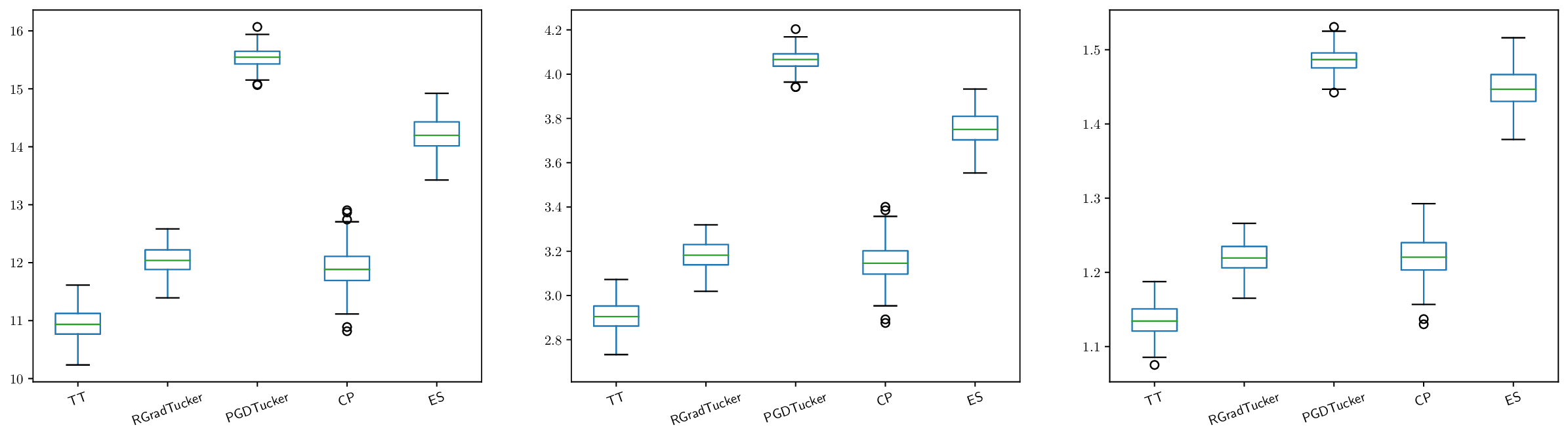}
	\caption{Boxplots of $\ell_1$ (left panel), $\ell_2$ (middle panel), $\ell_\infty$ (right panel) norm of forecasting errors for the ECOG example.}
	\label{ECOG:fig}
\end{figure}

\begin{figure}[ht]
	\centering 
	\includegraphics[width=1\textwidth,height=0.3\textwidth]{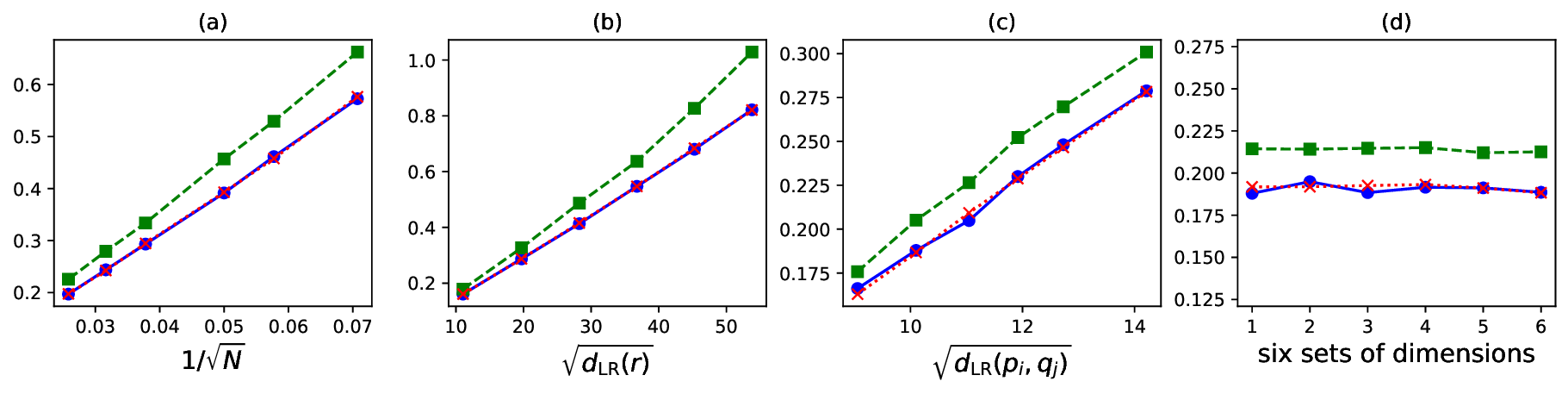}
	\caption{ Averaged estimation errors $\|\cm{\widehat{A}}_{\mathrm{LR}}-\cm{A}\|_{\mathrm{F}}$ for tensor train regression under Settings (a)-(d). Predictors and errors are independent with uniform distribution (\protect\Lratefive), independent with normality (\protect\Lrateten) or correlated with normality (\protect\Lratetwenty). Notation $d_{\mathrm{LR}}(r)$ refers to the values of $d_{\mathrm{LR}}$ with $r$ varying only, and similarly the notation $d_{\mathrm{LR}}(p_i,q_j)$. }
	\label{sim:regfig}
\end{figure}

\begin{figure}[ht]
	\centering
	\subfloat{\includegraphics[width=1.0\linewidth,height=0.4\linewidth]{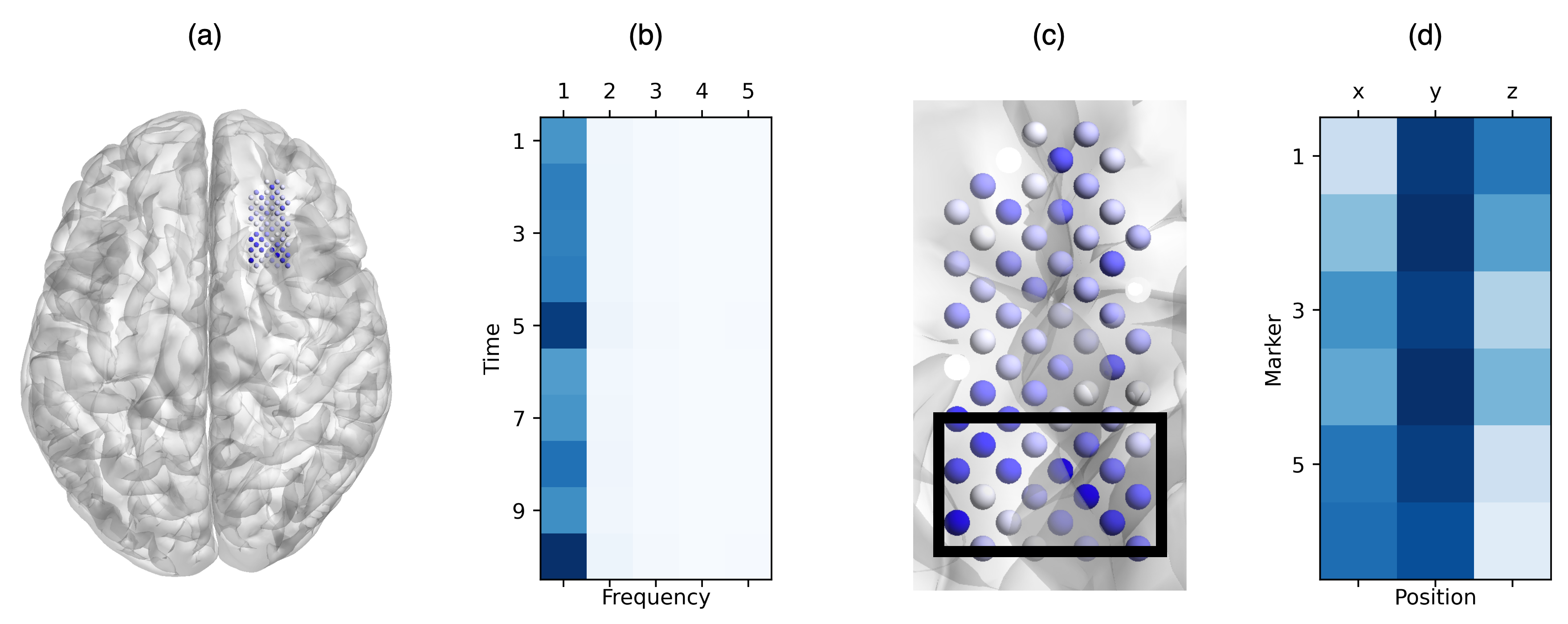}}
	\caption{Four plots from the ECoG example: (a) location the epidural space (blue area), (b) heatmap of absolute loadings on timestamps and center frequencies for the predictor factor with the mode of channels being averaged, (c) absolute loadings on 64 electrode channels for the high-frequency band factor from $\bm{G}_5$ and $\cm{G}_4$, and (d) heatmap of the absolute loading matrix for the response factor. Larger values are presented with darker colors, and the black bounding box at (c) corresponds to the primary motor cortex \citep{Shimoda2012}. }
	\label{brainfig3}
\end{figure}




\end{document}